\documentclass[11pt]{article}
\usepackage[T1]{fontenc}
\usepackage[utf8]{inputenc}
\usepackage{lmodern}
\usepackage{xspace}        
\usepackage{fullpage}
\usepackage{amsfonts,amsmath,amssymb, amsthm}
\usepackage{mfirstuc}
\def\withcolors{0}
\def\withnotes{0}
\def\fulldetails{1}
\def\withthmrestate{1}

\usepackage{dsfont} 
\usepackage{algorithmicx,algpseudocode,algorithm}
\usepackage[usenames,dvipsnames,table]{xcolor}
\usepackage[backref, colorlinks,citecolor=blue,bookmarks=true]{hyperref}
\usepackage{aliascnt}
\usepackage[numbered]{bookmark}
\usepackage[shortlabels]{enumitem}

 \usepackage[usenames,dvipsnames]{pstricks}
 \usepackage{epsfig}
 \usepackage{pst-grad} 
 \usepackage{pst-plot} 

 \setitemize{noitemsep,topsep=3pt,parsep=1pt,partopsep=1pt} 
 \setenumerate{itemsep=1pt,topsep=1pt,parsep=1pt,partopsep=1pt}
 \setdescription{itemsep=1pt}

\input{preamble}

\title{Sampling Correctors}
\author{
  \ccolor{Cl\'{e}ment L. Canonne}\thanks{Stanford University. Email: \email{ccanonne@cs.stanford.edu}. Supported by a Motwani Postdoctoral Fellowship. Part of this work was performed while the author was a graduate student at Columbia University, supported by NSF grants CCF-1115703 and NSF CCF-1319788; and an intern at Microsoft Research New England.}
  \and \tcolor{Themis Gouleakis}\thanks{CSAIL, MIT. Email: \email{tgoule@mit.edu}
Research supported by NSF grants CCF-1420692, CCF-1217423 and CCF-1065125.}
  \and \rcolor{Ronitt Rubinfeld}\thanks{CSAIL, MIT and the Blavatnik School of Computer Science, Tel Aviv University. Email: \email{ronitt@csail.mit.edu}. 
Research supported by ISF grant 1536/14 and NSF grants CCF-1420692, CCF-1217423 and CCF-1065125.}
}

\begin{document}
\setcounter{page}{0}
\pagenumbering{gobble}
\maketitle
\begin{abstract}
In many situations, sample data is obtained from a noisy or imperfect
source. In order to address such corruptions, this paper introduces the
concept of a \emph{sampling corrector}. Such algorithms use structure
that the distribution is purported to have, in order to allow one to
make ``on-the-fly'' corrections to samples drawn from probability
distributions. These algorithms then act as filters between the noisy
data and the end user.

We show connections between sampling correctors, distribution learning
algorithms, and distribution property testing algorithms. We show that
these connections can be utilized to expand the applicability of known
distribution learning and property testing algorithms as well as to achieve improved algorithms for those tasks.

As a first step, we show how to design sampling correctors using proper learning algorithms.
We then focus on the question of whether algorithms for sampling
correctors can be more efficient in terms of sample complexity than
learning algorithms for the analogous families of distributions. When
correcting monotonicity, we show that this is indeed the case 
when also granted query access to the cumulative distribution function. We also obtain sampling correctors for monotonicity even without this stronger type of access, provided that the distribution be originally \emph{very} close to monotone (namely, at a distance $O(1/\log^2 n)$). In addition to that, we consider a restricted error model that aims at capturing
``missing data'' corruptions. In this model, we show that distributions
that are close to monotone have sampling correctors that are
significantly more efficient than achievable by the learning approach.

We consider the question of whether an additional source of independent
random bits is required by sampling correctors to implement the
correction process. We show that for correcting close-to-uniform
distributions and close-to-monotone distributions, no additional source
of random bits is required, as the samples from the input source itself
can be used to produce this randomness.

\end{abstract}
 
\clearpage
\setcounter{tocdepth}{2}  \tableofcontents
\clearpage
\pagenumbering{arabic}

\section{Introduction}

Data consisting of samples from distributions 
is notorious for reliability issues:  Sample
data can be greatly affected by noise, calibration
problems or other faults in the sample recording process;
portions of data may be lost; 
extraneous samples may be erroneously recorded.
Such noise may be completely random, or may have
some underlying structure. To give a sense
of the range of difficulties one might have
with sample data, we mention some examples:
A sensor network which tracks traffic data
may have dead sensors  which transmit
no data at all, or other sensors that are defective
and transmit arbitrary numbers. Sample data
from surveys may suffer from response  rates
that are correlated  with location or socioeconomic
factors. Sample data from species
distribution models are prone to
geographic location errors~\cite{HBTB:14}.

Statisticians have grappled with defining a methodology
for working with
distributions in the presence of noise by {\em correcting} 
the samples. If, for example, you know that the uncorrupted distribution is
Gaussian, then it would be natural to correct the samples
of the distribution to  the nearest Gaussian.   
The challenge in defining this methodology is:
how do you correct the samples if you do not
know much about the original uncorrupted distribution?   
To analyze distributions with noise in a principled way,
approaches have included \emph{imputation} \cite{book:imputation,book:incomplete:data,SMP:07} for the case of missing or incomplete data, and \emph{outlier detection and removal} \cite{book:outliers,article:outliers,book:outliers:detect} to handle ``extreme points'' deviating significantly from the underlying distribution. 
More generally, the question of coping with the \emph{sampling bias} inherent to many strategies (such as opportunity sampling) used in studying rare events or species, or with inaccuracies in the reported data, is a key challenge in many of the natural and social sciences 
(see e.g.~\cite{Signor:1982,Senese:03,Scholarpedia:2008}). 
While these problems are usually dealt with drawing on additional knowledge or by using specific modeling assumptions, 
no general procedure is known that addresses them in a systematic fashion.

In this work, we propose a  methodology which is 
based on using {\em known structural properties} of the distribution
to design {\em sampling correctors} which
``correct'' the sample data. 
While assuming these structural properties is in itself a type of modeling, it is in general much weaker than postulating a strict form of the data (e.g., that it follows a linear model perturbed by Gaussian noise).
Examples of structural properties  which might be used to
correct samples include the property of being
bimodal, a mixture of several Gaussians, a mixture of piecewise-polynomial distributions, or an independent joint distribution.
Within this methodology, the main
question is: how best can one output samples of a distribution in which
on one hand, the structural properties are restored, and on the other hand,
the corrected distribution is close to the original distribution?
We show that this task is intimately connected
to distribution learning tasks, 
but we also give instances in which
such tasks can be performed strictly more efficiently.

\subsection{Our model}

We introduce two (related) notions of algorithms to correct distributions: 
\emph{sampling correctors}
and \emph{sampling improvers}.
Although the precise definitions are deferred to \autoref{sec:definitions}, 
we describe and state informally what we mean by these. 
In what follows, \domain is a finite
domain, \property is any fixed property of 
distributions, i.e., a subset of distributions,  over \domain
and distances between distributions are measured according
to their \emph{total variation distance}.\footnote{The total variation
distance is defined as
 $\totalvardist{\D_1}{\D_2} 
\eqdef \max_{S\subseteq \domain}\left( \D_1(S)-\D_2(S) \right) 
= \frac{1}{2}\sum_{x\in \domain} \abs{\D_1(x) - \D_2(x)}.$
}

\medskip 

A \emph{sampling corrector} for \property 
is a randomized algorithm which gets samples from 
a distribution $\D$ guaranteed to be \eps-close to having property
\property, and outputs a sample from a ``corrected distribution'' $\tilde{\D}$ 
which, with high probability, \textsf{(a)} \emph{has} the property; and \textsf{(b)} 
is still close to the original distribution $\D$ (i.e., within distance $\eps_1$). The 
\emph{sample complexity} of such a corrector is the number of samples it needs to 
obtain from $\D$ in order to output one from $\tilde{\D}$. 

\new{
To make things concrete, we give a simple example of correcting independence of distributions over a product space $[n]\times[m]$. 
For each pair of samples $(x,y)$ and $(x^\prime, y^\prime)$ from a distribution $\D$ which is \eps-close to
independent, output \emph{one} sample $(x,y^\prime)$. As $x$ and $y^\prime$ are independent, the resulting distribution clearly has the property; and it can be shown that if $\D$ was indeed $\eps$-close to independent, then the distribution of $(x,y^\prime)$ will indeed be $3\eps$-close to $\D$~\cite{SV:98}. (Whether this sample complexity can be reduced further to $q < 2$, even on average, is an open question.)
}

Note that in some settings it may be too much to ask for complete correction 
(or may even not be the most desirable option). 
For this reason, we also consider the weaker notion of
\emph{sampling improvers}, 
which is similar to a sampling corrector but is only required to
transform the distribution into a new distribution which is 
{\em closer} to having the property \property.

\smallskip

One naive way to solve these problems, the ``learning approach,''
is to approximate the probability mass function of $\D$, 
and find a candidate $\tilde{\D} \in \property$.
Since we assume we have a complete description
of $\tilde{\D}$, we can then output samples according to $\tilde{\D}$ 
without further access to $\D$.  
In general, such an approach can be very inefficient in terms of time complexity.
However, if there is 
an {\em efficient} agnostic proper learning algorithm\footnotemark{} for $\property$, 
we show that this approach can lead to efficient sampling correctors.
For example, we use such an approach to give sampling correctors
for the class of monotone distributions.
\smallskip
\footnotetext{\label{ftn:learning}Recall that 
a \emph{learning algorithm} for a 
class of distributions \class is an algorithm which 
gets independent samples from an unknown distribution $\D\in\class$; 
and on input \eps must, with high probability, 
output a hypothesis which is \eps-close to $\D$ in total variation distance. 
If the hypotheses the algorithm produces are guaranteed to belong to \class as well, 
we call it a \emph{proper learning algorithm}. 
Finally, if the~--~not-necessarily proper~--~algorithm is able to learn 
distributions that are only \emph{close} to \class, 
returning a hypothesis at a distance at most $\opt+\eps$ from $\D$ 
-- where $\opt$ is the distance from $\D$ to the class, 
it is said to be \emph{agnostic}.
For a formal definition of these concepts, 
the reader is referred to~\autoref{ssec:connections:agnostic} and~\autoref{appendix:definitions}.}

In our model, we wish to optimize the following 
two parameters of our correcting algorithms:
The first parameter 
is the number of samples of $\D$ needed to output samples of $\tilde{\D}$. The
second parameter is
the number of {\em additional} truly random bits 
needed for outputting samples of $\tilde{\D}$.
Note that in the above learning approach, the dependence on each
of these parameters could be quite large.
Although these parameters are not independent of each
other (if $\D$ is of high enough entropy, then it can be used
to simulate truly random bits), they can be thought of as complementary, as one typically will aim at a tradeoff between the two. 
Furthermore, a parsimonious use of extra random bits may be crucial 
for some applications,
while in others the correction of the data itself is the key factor; 
for this reason, 
we track each of the parameters separately.
For any property $\property$, the main question is whether one
can achieve improved complexity  in terms of
these parameters over the use of the naive (agnostic) learning approach
for $\property$.

\subsection{Our results}

Throughout this paper, we will focus on two particular properties of interest, 
namely \emph{uniformity} and \emph{monotonicity}. 
The first one, arguably one of the most natural and illustrative properties to be considered, 
is nonetheless deeply challenging in the setting of randomness scarcity. As for the second, 
not only does it provide insight in the workings of sampling correctors as well as non-trivial connections and algorithmic results, 
but is also one of the most-studied classes of distributions in the statistics and probability literature, with a body of work covering several decades (see 
e.g.~\cite{Grenander:56,Birge:87,BKR:04,DDS:12}, or \cite{DDSVV:13} for a detailed list of references). 
Moreover, recent work on distribution testing~\cite{DDSVV:13,CDGR:15} shows strong 
connections between monotonicity and a wide range of other properties, 
such as for instance log-concavity, Monotone Hazard Risk and Poisson Binomial Distributions. 
This gives evidence that the study of monotone distributions may have direct implications for 
correction of many of these ``shape-constrained properties.''

\paragraph{Sampling correctors, learning algorithms and property testing algorithms.}
We begin by showing implications of the existence of
sampling correctors for the existence of various types of learning and property
testing algorithms in other 
models. We first show in \autoref{theo:connection:regular:learning} that learning algorithms for a distribution class
imply sampling correctors for distributions in this class (under
{\em any} property to correct) with the same sample complexity,
though not necessarily the same running time dependency.  
However,  when efficient agnostic proper
learning algorithms for a distribution class
exist, {we show that} there are efficient sampling correctors 
for the same class.
In~\cite{Birge:87,CDSS:14} efficient algorithms for agnostic learning of
concise representations
for several families of distributions are given, including distributions
that are monotone, $k$-histograms, Poisson binomial, and sums of $k$ independent
random variables.  Not all of these algorithms are proper. 

Next, we show in~\autoref{theo:connection:agnostic:learning} that the existence of \textsf{(a)} an efficient
learning algorithm, as e.g. in \cite{ILR:12,CDSS:13,DDS:PBD:12,DDOST:13}, and \textsf{(b)} an efficient sampling
corrector for a class of distributions implies an efficient 
{\em agnostic} learning
algorithm for the same class of distributions.
It is well known that agnostic learning 
can be much harder than non-agnostic learning, 
as in the latter the algorithm is able to leverage structural properties of the class \class. \new{Thus, by the above result we also get that any agnostic learning lower bounds can be used to obtain sampling corrector lower bounds.}

Our third result in this section, \autoref{theo:connection:corrector:testing}, shows that an efficient property tester, 
an efficient distance estimator (which computes an additive estimate of the distance between two distributions) and an efficient sampling corrector for
a distribution class imply a tolerant property tester with complexity
equal to the complexity of correcting the number of samples required to run both the tester and estimator.\footnote{Recall that the difference between  testing and tolerant testing lies in that the former asks to distinguish whether an unknown distribution \emph{has} a property, or is far from it, while the latter requires to decide whether the distribution is \emph{close} to the property versus far from it. (See~\autoref{appendix:definitions} for the rigorous definition.)}
As tolerant property testing can be much more difficult than
property testing \cite{GRexp:00,BFRSW:10,Paninski:08,ValiantValiant:11}, this gives a general purpose way of getting
both upper bounds on tolerant property testing and lower bounds
on sampling correctors.\medskip

\noindent We describe how these results can be employed in \autoref{sec:connections}, where we give 
specific applications in achieving improved property testers for various properties.

\paragraph{Is sampling correction easier than learning?}
We next turn to the question of whether there are natural examples of sampling correctors whose query complexity is
asymptotically smaller 
than that of distribution learning algorithms for the same class. 
While the sample complexity of learning monotone 
distributions is known to be $\bigOmega{\log{n}}$~\cite{Birge:87} 
(this lower bound on the sample and query complexity
holds even when the algorithm is allowed both to make
queries to the cumulative distribution function
as well as to access samples of the distribution), 
we present in \autoref{sec:monotonicity:oblivious} an oblivious sampling corrector for 
monotone distributions whose sample complexity is $\bigO{1}$ {and} that corrects error that is smaller than $\eps \leq \bigO{1/\log^2{n}}$. 
This is done by first implicitly approximating the distribution by a ``histogram'' on only a small number of intervals, 
using ingredients from \cite{Birge:87}. 
This (very close) approximation can 
then be combined, still in an oblivious way, with a carefully chosen slowly decreasing distribution, so that the resulting 
mixture is not only guaranteed to be monotone, but also close to the original distribution.

It is open whether there exist sampling correctors for monotone distributions
with sample complexity $\littleO{(\log n)/\eps^3}$ that can correct arbitrary error $\eps\in(0,1)$, thus beating the sample complexity of the ``learning approach.'' (We note however that a logarithmic dependence on $n$ is inherent \new{when $\eps = \omega(1/\log n)$}, as pointed out to us by Paul Valiant~\cite{Val:15:sketch}.)

Assuming a stronger type of access to the unknown distribution -- namely, query access to its cumulative distribution function (\cdf) as in \cite{BDKR:05,CR:14}, we 
describe in~\autoref{sec:monotonicity:cdf} a sampling corrector for monotonicity with (expected) query complexity $\bigO{\sqrt{\log n}}$ {which works for arbitrary $\eps\in(0,1)$}. At a high-level, our 
algorithm combines the ``succinct histogram'' technique mentioned above with a two-level bucketing approach to correct the distribution first at a very coarse level 
only (on ``superbuckets''), 
and defer the finer corrections (within a given superbucket) to be made on-the-go at query time. 
The challenge in this last part is that
one must ensure that all of these disjoint local 
 corrections are consistent with each other -- 
and crucially, \emph{with all future sample corrections.} 
 {To achieve this, we use a ``boundary correction'' subroutine which fixes potential violations between two neighboring superbuckets by evening out the boundary differences. To make it possible, we use rejection sampling to allocate adaptively an extra ``budget'' to each superbucket that this subroutine can use for corrections.}

\paragraph{Restricted error models.}
Since many of the sampling correction problems are 
difficult to solve in general, we suggest error models
for which more efficient sampling correction algorithms
may exist. A first class of error models, which we refer to
as {\em missing data errors}, is introduced in~\autoref{sec:specific:errors} and defined as follows -- given a distribution
over $[n]$, all samples in some interval $[i,j]$ for $1<i<j<n$ are
deleted.  Such errors could correspond to samples from a sensor
network where one of the sensors ran out of power; 
emails mistakenly deleted by a spam filter; or samples
from a study in which some of the paperwork got lost.
Whenever the input distribution $\D$, 
whose distance from monotonicity is {$\eps\in(0,1)$}, 
falls under this model, we give a sampling improver that is able to find 
a distribution both $\eps_2$-close to monotone
 and $\bigO{\eps}$-close to the original 
using $\tildeO{{1}/{\eps_2^{3}}}$ samples. 
The improver works in two stages. In the ``preprocessing stage,'' 
 we detect the location of the missing interval 
(when the missing weight is sufficiently large) 
and then estimate its missing weight, using a 
``learning through testing'' approach from \cite{DDS:12} to keep the sample complexity under control. 
In the second stage, 
we give a procedure 
by which the algorithm can use its knowledge
of the estimated missing interval to
 correct the distribution by rejection sampling.
 
 \paragraph{Randomness Scarcity.}
We {then} consider the case where only a limited amount of randomness (other than the input distribution) is available, and optimizing its use, possibly at the cost of worse parameters and/or sample complexity of our sampling improvers, is crucial. This captures situations where generating the random bits the algorithm use is either expensive\footnote{On this topic, see for instance the discussion in \cite{KR:94,IZ:89}, and references therein.} (as in the case of physical implementations relying on devices, such as Geiger counters or Zener diodes) or undesirable  (e.g., when we want the output distribution to be a deterministic function of the input data, for the sake of reproducibility or parallelization). We focus on this setting in \autoref{sec:focus:randomness}, and provide sampling correctors and improvers for uniformity that use samples \emph{only} from the input distribution. For example, we give a sampling improver that, given access to distribution $\eps$-close to uniform, grants access to a distribution $\eps_2$-close to uniform distribution and has \emph{constant} sample complexity $O_{\eps, \eps_2}(1)$. We achieve this by exploiting the fact that the uniform distribution is not only an absorbing element for convolution in Abelian groups, but also an \emph{attractive fixed point} with high convergence rate. That is, by convolving a distribution with itself (i.e., summing independent samples modulo the order of the group) one gets very quickly close to uniform. Combining this idea with a different type of improvement (based on a von Neumann-type ``trick'') allows us to obtain an essentially optimal tradeoff between closeness to uniform and to the original distribution.

\subsection{Open problems}

\paragraph*{Correcting vs. Learning}
A main direction of
interest would be to obtain more examples of 
properties for which correcting is strictly more efficient than 
(agnostic or non-agnostic) learning.   Such examples would be
insightful even if they are more efficient only
in terms of the number of samples 
required from the original distribution, 
without considering the additional randomness requirements for
generating the distribution.
More specifically, one may ask whether 
there exists a 
sampling corrector for 
monotonicity of 
distributions (i.e., one that beats the learning bound from \autoref{lemma:learning:corrector:birge})
for all $\epsilon<1$ which uses at most
$\littleO{(\log n)/\eps^3}$ samples from the original distribution per sample
output of the corrected distribution. \new{Other properties of interest, among many, include log-concavity of distributions, having a piecewise-constant density (i.e., being a $k$-histogram for some fixed value $k$), or being a Poisson Binomial Distribution.}

\paragraph*{The power of additional queries}
Following the line of work pursued in \cite{CFGM:13,CRS:14,CR:14} (in the setting of distribution testing), 
it is natural in many situations to consider additional types of queries to the input distribution: 
e.g., either \emph{conditional queries} 
(getting a sample conditioned on a specific subset of the domain) or \emph{cumulative queries} 
(granting query access to the cumulative distribution function, 
besides the usual sampling). 
By providing algorithms with this extended access to the underlying probability distribution, 
can one obtain faster sampling correctors for specific properties, as we \new{do} in~\autoref{sec:monotonicity:cdf} in the case of monotonicity?

\paragraph*{Confidence boosting}
Suppose that there exists, for some property \property, a sampling improver \Algo that only guarantees a success probability\footnote{We note that the case of interest here is of batch sampling improvers: indeed, in order to generate a single draw, a sampling improver acts in a non-trivial way only if the parameter $\eps$ is greater than its failure probability $\delta$. If not, a draw from the original distribution already satisfies the requirements.} of $2/3$. Using \Algo as a black-box, can one design a sampling improver $\Algo^\prime$ which succeeds with probability $1-\delta$, for any $\delta$?

More precisely, let $\Algo$ be a batch improver for \property which, when queried, makes $q(\eps_1,\eps_2)$ queries and provides $t\geq 1$ samples, with success probability at least $2/3$. Having black-box access to $\Algo$, can we obtain a batch improver $\Algo^\prime$ which on input $\delta>0$ provides $t^\prime \geq 1$ samples, with success probability at least $1-\delta$? If so, what is the best $t^\prime$ one can achieve, and what is the minimum query complexity of $\Algo^\prime$ one can get (as a function of $q(\cdot,\cdot)$, $t^\prime$ and $\delta$)?

This is known for property testing (by running the testing algorithm independently $\bigO{\log(1/\delta)}$ times and taking the majority vote), 
as well as for learning (again, 
by running the learning algorithm many times, 
and then doing hypothesis testing, e.g.\xspace\emph{\`a la}  
\cite[Theorem 19]{DK:13}). 
However, these approaches do not straightforwardly 
generalize to sampling improvers or correctors, 
respectively because the output is not a single bit, 
and as we only obtain a sequence of samples (instead of an actual, 
fully-specified hypothesis distribution).

\subsection{Previous work}
Dealing with noisy or incomplete datasets has been a challenge in Statistics and data sciences, and many methods 
have been proposed to handle them. One of the most widely used, \emph{multiple imputation} (one of many variants of the 
general paradigm of \emph{imputation}) was first introduced by Rubin~\cite{Rubin:87} and consists of the creation of several complete datasets from an incomplete one. 
Specifically, one first obtains these new datasets by filling in the missing values randomly according to a maximum likelihood (ML) distribution computed from the observations and a modeling assumption made on the data. The parameters of this model are then updated using the new datasets and the ML distribution is computed again. This resembles the Expectation-Maximization (EM) algorithm, \new{a heuristic} which can also be used for similar problems, as e.g.~in \cite{DLR77}. \new{In cases where the EM algorithm does converge to the ML distribution,} after a few iterations one can get both accurate parameter estimates and the right distribution to sample data from. Assuming the assumptions chosen to model 
the data did indeed reflect its true 
 distribution, and that the number of these new datasets was large enough, this can be shown to yield statistically accurate and unbiased 
 results~\cite{book:incomplete:data,book:imputation}.

From a Theoretical Computer Science perspective, the problem of local correction of data
has received much attention in the contexts
of self-correcting programs, locally correctable codes, and
local filters for graphs and functions over $[n]^d$ 
(some examples include 
\cite{BLR:90,book:yekhanin:ldc,ACCL:08,SS:10,BGJRW:12,JR:11}). 
To {the best of} our knowledge, this is the first work to address the
correction of data from distributions. (We observe that Chakraborty et al. consider in \cite{CGM:11} a different question, although of 
a similar distributional flavor: namely, given query access to a Boolean function $f\colon\{0,1\}^n\to\{0,1\}$ which is close to a $k$-junta $f^\ast$, they show how 
to approximately generate uniform PAC-style samples of the form $\langle x,g^\ast(x)\rangle$ where $x\in\{0,1\}^k$ and $g^\ast$ is the function underlying $f^\ast$. They then describe how to apply this ``noisy sampler'' primitive to test whether a function is close to being a junta.)

In this work, we show that the problem of estimating distances
between distributions is related. There has been much work on this
topic, but we note the following result:
\cite{DDSVV:13} show how to estimate the total 
variation distance between $k$-modal probability distributions.\footnote{A probability distribution $\D$ is \emph{$k$-modal} if there exists a partition of $[n]$ in $k$ intervals such that $\D$ is monotone (increasing or decreasing) on each.}
The authors give a reduction of their problem into one with logarithmic size, 
using a result by Birg\'e on monotone distributions~\cite{Birge:87}. In particular, 
one can partition the domain $\domain=[n]$ into $\log n/\eps$ intervals in a oblivious way, 
such that the ``flattening'' of any monotone distribution according 
to that interval is $\bigO{\eps}$-close to the original one. 
We use similar ideas in order to obtain some of the results in the present paper.

\new{Another related field in Statistics is that of \emph{robust statistics}, which is concerned with estimating the parameters of a model in spite of a fraction of the data being corrupted (equivalently, under some model misspecification). This type of question is similar in spirit to our setting, in that it aims at overcoming noisy or corrupted data; however, the focus there is on \emph{learning} characteristics of the purported model, instead of ``removing'' the noise (which could be seen as a less demanding goal, and a possible approach towards the learning task itself). Well-studied in Statistics since the seminal work of Tukey~\cite{Tukey:60}, robust statistics have recently been studied from a computational and algorithmic viewpoint: we refer the reader to~\cite{HR:09,HRRS:11} for surveys and overviews of the field itself, and to~\cite{LRV:16,DKKLMS:16,CSV:17,DKKLMS:17,BDLS:17,DKKLMS:18} for recent advances from a theoretical computer science perspective.

Finally, it} is instructive to compare the goal of our model of distribution
sampling correctors to that of extractors: in spite of many similarities, the two have essential differences and the results are in many cases incomparable. We defer this discussion to \autoref{sec:uniformity}.
 
\section{Preliminaries}\label{sec:preliminaries}
Hereafter, we write $[n]$ for the set $\{1,\dots,n\}$, and $\log$ for the logarithm in base $2$. A \emph{probability distribution} over a finite domain $\domain$ is a non-negative function $\D\colon\domain\to[0,1]$ such that $\sum_{x\in\domain} \D(x) = 1$; we denote by $\uniform_{\domain}$ the uniform distribution on \domain. Moreover, given a distribution $\D$ over $\domain$ and a set $S\subseteq\domain$, we write $\D(S)$ for the total probability weight $\sum_{x\in S} \D(x)$ assigned to $S$ by $\D$.

\paragraph{Previous tools from probability.}
As previously mentioned, in this work we will be concerned with the total variation distance between distributions. Of interest for the analysis of some of our algorithms, and assuming \domain is totally ordered (in our case, $\domain=[n]$), one can also define the \emph{Kolmogorov distance} between $\D_1$ and $\D_2$ as
\begin{equation}\label{eq:def:dk}
  \kolmogorov{\D_1}{\D_2} \eqdef \max_{x\in \domain}\abs{F_1(x)-F_2(x)} \end{equation}
where $F_1$ and $F_2$ are the respective cumulative distribution functions (\cdf) of $\D_1$ and $\D_2$. Thus, the Kolmogorov distance is the $\lp[\infty]$ distance between the \cdf's; and $\kolmogorov{\D_1}{\D_2} \leq \totalvardist{\D_1}{\D_2} \in [0,1]$. \smallskip

We first state the following theorem, which guarantees that for any two distributions $\D_1$, $\D_2$, applying any (possibly randomized) function to both $\D_1$ and $\D_2$ can never increase their total variation distance:
\begin{fact}[Data Processing Inequality for Total Variation Distance]\label{lemma:data:processing:inequality:total:variation}
Let $\D_1$, $\D_2$ be two distributions over a domain $\Omega$. Fix any randomized function\footnote{Which can be seen as a distribution over functions over $\Omega$.} $F$ on $\Omega$, and let $F(\D_1)$ be the distribution such that a draw from $F(\D_1)$ is obtained by drawing independently $x$ from $\D_1$ and $f$ from $F$ and then outputting $f(x)$ (likewise for $F(\D_2)$).
Then we have
\[
\totalvardist{ F(\D_1) }{  F(\D_2) }  \leq \totalvardist{ \D_1 }{ \D_2 }. \]
\end{fact}

Finally, we recall below a fundamental fact from probability theory that will be useful to us, the \emph{Dvoretzky--Kiefer--Wolfowitz (DKW) inequality}. Informally, this result says that one can learn the cumulative distribution function of a distribution up to an additive error $\eps$ in $\lp[\infty]$ distance, by taking only $\bigO{1/\eps^2}$ samples from it.
\begin{theorem}[\cite{DKW:56,Massart:90}]\label{theo:dkw}
Let $\D$ be a distribution over $[n]$. Given $m$ independent samples $x_1,\dots ,x_m$ from $\D$, define the empirical distribution $\hat{\D}$ as follows:
\[
\hat{\D}(i) \eqdef \frac{\abs{ \setOfSuchThat{j\in[m]}{x_j=i} } }{m}, \quad i\in[n].
\]
Then, for all $\eps > 0$, $\probaOf{ \kolmogorov{\D}{\hat{\D}} > \eps } \leq 2e^{-2m\eps^2}$, where the probability is taken over the samples.
\end{theorem} 
\noindent In particular, setting $m=\bigTheta{\frac{\log(1/\delta)}{\eps^2}}$ we get that $\kolmogorov{\D}{\hat{\D}} \leq \eps$ with probability at least $1-\delta$.

\newer{\paragraph{Flattenings.}
\noindent For a distribution $\D$ and a partition of $[n]$ into intervals $\mathcal{I}=(I_1,\dots,I_\ell)$, we define the \emph{flattening of $\D$ with relation to $\mathcal{I}$} as the distribution $\Psi_{\mathcal{I}}(\D)$, where $\Psi_{\mathcal{I}}(\D)(i) = {\D(I_k)}/{\abs{I_k}}$ for all $k\in [\ell]$ and $i\in I_k$. A straightforward computation (see \autoref{appendix:misc:proofs}) shows that such flattening cannot increase the distance between two distributions, i.e.
  \begin{equation}\label{eq:Birge:tv}
      \totalvardist{ \Psi_{\mathcal{I}}(\D_1) }{ \Psi_{\mathcal{I}}(\D_2) } \leq \totalvardist{\D_1}{\D_2}.
  \end{equation}
  }

\paragraph{Monotone distributions.}
We say that a distribution $\D$ on $[n]$ is \emph{monotone} (non-increasing) if its probability mass function is non-increasing, that is if $\D(1)\geq \dots \geq \D(n)$. 
When dealing with monotone distributions, it will be useful to consider the \emph{Birg\'e decomposition}, which is a way to approximate any monotone distribution $\D$ by a histogram, where the latter is supported by logarithmically many intervals \emph{which crucially do not depend on $\D$ itself}:
\begin{definition}[Birg\'e decomposition]\label{def:birge:decomposition}
  Given a parameter $\alpha>0$, the corresponding (oblivious) \emph{Birg\'e decomposition of $[n]$} is the partition $\mathcal{I}_\alpha=(I_1,\dots,I_\ell)$, where $\ell=\bigTheta{\frac{\ln( \alpha n + 1)}{\alpha}}=\bigTheta{\frac{\log n }{\alpha}}$ and $\abs{I_{k}}=\flr{(1+\alpha)^k}$, $1\leq k \leq \ell$. 
\end{definition}

\noindent \newer{For a distribution $\D$ and parameter $\alpha$, define $\birge[\D]{\alpha}$ to be the ``flattened'' distribution with relation to the oblivious decomposition $\mathcal{I}_\alpha$, that is $\birge[\D]{\alpha} = \Psi_{\mathcal{I}_\alpha}(\D)$.} 
The next theorem states that every monotone distribution can be well-approximated by its flattening on the Birg\'e decomposition's intervals:     
\begin{theorem}[\cite{Birge:87,DDSVV:13}]\label{theorem:Birge:obl:decomp}
 If $\D$ is monotone, then $\totalvardist{\D}{\birge[\D]{\alpha}} \leq \alpha$.
\end{theorem}

\noindent As a corollary, one can extend the theorem to distributions only promised to be \emph{close} to monotone:
\begin{restatable}{corollary}{birgerobustcorollary}\label{coro:Birge:decomposition:robust}
    Suppose $\D$ is \eps-close to monotone, and let $\alpha > 0$. Then $\totalvardist{\D}{ \birge[\D]{\alpha} } \leq 2\eps + \alpha$. Furthermore,  $\birge[\D]{\alpha}$ is also \eps-close to monotone.
\end{restatable}
  
\paragraph{Access to the distributions.} While we will mostly be concerned in this work with the standard model of access to the probability distributions, where the algorithm is provided with independent samples from an unknown distribution $\D$, the concepts we introduce and some of our results apply to some other types of access as well. One in particular, the \Cdfsamp access model, grants the algorithms the ability to query the value of the cumulative distribution function (\cdf) of $\D$, in addition to regular sampling.\footnote{See also \cite{clement:survey:distributions} for a summary and comparison of the different existing access models.} (We observe, as in \cite{CR:14}, that this type of query access is for instance justified when the distribution originates from a sorted dataset, in which case such queries can be implemented with only a logarithmic overhead.)

Unless explicitly specified otherwise, our algorithms only assume standard sampling access; the formal definitions of the two models mentioned above can be found in~\autoref{appendix:definitions}.
 
\section{Our model: definitions}\label{sec:definitions}
In this section, we state the precise definitions of sampling correctors, improvers and batch sampling improvers. 
To get an intuition, 
the reader may think for instance of the parameter $\eps_1$ below as being $2\eps$, 
and the error probability $\delta$ as $1/3$.
Although all definitions are presented in terms of the total variation
distance, analogous definitions in terms of other distances can also be made.

\begin{definition}[Sampling Corrector]\label{def:sampling:corrector}
  Fix a given property $\property$ of distributions on $\domain$. 
An \emph{$(\eps,\eps_1)$-sampling corrector for $\property$} is a randomized algorithm which is given parameters $\eps, \eps_1\in(0,1]$ such that $\eps_1 \geq \eps$ and $\delta\in[0,1]$, as well as sampling access to a distribution $\D$. Under the promise that $\totalvardist{\D}{\property}\leq \eps$, the algorithm must provide, with probability at least $1-\delta$ over the samples it draws and its internal randomness, sampling access to a distribution $\tilde{\D}$ such that
  \begin{enumerate}[(i)]
    \item $\tilde{\D}$ is close to $\D$: $\totalvardist{\tilde{\D}}{\D} \leq \eps_1$;
    \item $\tilde{\D}$ has the property: $\tilde{\D}\in \property$.
  \end{enumerate}
  In other terms, with high probability the \new{algorithm}   will simulate exactly a sampling oracle for $\tilde{\D}$. The query complexity $q=q(\eps,\eps_1,\delta,\domain)$ of the algorithm is the number of samples from $\D$ it takes per query in the worst case.
\end{definition}

One can define a more general notion, which allows the algorithm to only get ``closer'' to the desired property, and convert some type of access $\ORACLE_1$ into some other type of access $\ORACLE_2$ (e.g., from sampling to evaluation access):
\begin{definition}[Sampling Improver (general definition)]\label{def:sampling:corrector:general}
  Fix a given property $\property$ over distributions on $\domain$. A \emph{sampling improver for $\property$} (from $\ORACLE_1$ to $\ORACLE_2$) is a randomized algorithm which, given parameter $\eps\in(0,1]$ and $\ORACLE_1$ access to a distribution $\D$ with the promise that $\totalvardist{\D}{\property}\leq \eps$ as well as parameters $\eps_1,\eps_2\in[0,1]$ satisfying $\eps_1+\eps_2\geq \eps$, provides, with probability at least $1-\delta$ over the answers from $\ORACLE_1$ and its internal randomness, $\ORACLE_2$ access to a distribution $\tilde{\D}$ such that
  \begin{align*}
    \totalvardist{\tilde{\D}}{\D} \leq \eps_1 \tag{Close to $\D$}\\
    \totalvardist{\tilde{\D}}{\property} \leq \eps_2 \tag{Close to \property}
  \end{align*}
  In other terms, with high probability the algorithm
   will simulate exactly $\ORACLE_2$ access to $\tilde{\D}$. The query complexity $q=q(\eps,\eps_1,\eps_2,\delta,\domain)$ of the algorithm is the number of queries it makes to $\ORACLE_1$ in the worst case.
\end{definition}

Finally, one may ask for such an improver to provide \emph{many} samples from the (same) improved distribution,\footnote{Indeed, observe that as sampling correctors and improvers are randomized algorithms with access to their ``own'' coins, there is no guarantee that fixing the input distribution $\D$ would lead to the same output distribution $\tilde{\D}$. This is particularly important when providing other types of access (e.g., evaluation queries) to $\tilde{\D}$ than only sampling.} where ``many'' is a number committed in advance. We refer to such an algorithm as a \emph{batch sampling improver} \new{(or, similarly, batch sampling corrector):}
\begin{definition}[Batch Sampling Improver]\label{def:sampling:corrector:batch}
  For $\property$, $\D$, $\eps$, $\eps_1,\eps_2\in[0,1]$ as above, and parameter $m\in\N$, a \emph{batch sampling improver for $\property$} (from $\ORACLE_1$ to $\ORACLE_2$) is a sampling improver which provides, with probability at least $1-\delta$, $\ORACLE_2$ access to $\tilde{\D}$ for as many as $m$ queries, in between which it is allowed to maintain some internal state ensuring consistency. The query complexity $q$ of the algorithm is now allowed to depend on $m$ as well, \new{i.e. $q=q(\eps,\eps_1,\eps_2,m,\delta,\domain)$.}
  
  \noindent Note that, in particular, when providing sampling access to $\tilde{\D}$ the batch improver must guarantee independence of the $m$ samples. When $\eps_2$ is set to 0 in the above definition, we will refer to the algorithm as a \emph{batch sampling corrector}, \new{with query complexity $q(\eps,\eps_1,m,\delta,\domain)$.}
\end{definition}

\begin{remark}[On parameters of interest.]
We observe that the regime of interest of our correctors and improvers is when the number of corrected samples to output is at least of the order $\bigOmega{1/\eps}$. Indeed, if fewer samples are required, then the assumption that the distribution $\D$ be \eps-close to having the property implies that -- with high probability -- a small number of samples from $\D$ will be indistinguishable from the closest distribution having the property. (So that, intuitively, they are already ``as good as it gets,'' and need not be corrected.)
\end{remark}
\begin{remark}[On testing lower bounds]\label{rk:samples:lb:pt}
A similar observation holds for properties \property that are known to be \emph{hard to test}, that is for which some lower bound of $q(n,\eps)$ samples holds to decide whether a given distribution satisfies \property, or is \eps-far from it. In light of such a lower bound, one may wonder whether there is something to be gained in correcting $m< q(n,\eps)$ samples, instead of simply using $m$ samples from the original distribution altogether. However, such a result only states that there exists \emph{some} worst-case instance $\D^\ast$ that is at distance $\eps$ from the property \property, yet requires this many samples to be distinguished from it: so that any algorithm relying on samples from distributions satisfying \property could be fed $q(n,\eps)-1$ samples from this particular $\D^\ast$ without complaining. Yet, for ``typical'' distributions that are $\eps$-close to \property, far fewer samples are required to reveal their deviation from it: for many, as few as $O(1/\eps)$ suffice. Thus, an algorithm that expects to get say $q(n,\eps)^{.99}$ samples from a honest-to-goodness distribution from \property, but instead is provided with samples from one that is merely $\eps$-close to it, may break down very quickly. Our corrector, in this very regime of $\littleO{q(n,\eps)}$ samples, guarantees this will not happen.
\end{remark}

We conclude this section by introducing a relaxation of the notion of sampling corrector, where instead of asking the unknown distribution be close to the class it is corrected for we instead decouple the two. For instance, one may require the unknown distribution to be close to a Binomial distribution, but only correct it to be unimodal. This leads to the following definition of a \emph{non-proper corrector}:
\begin{definition}[Non-Proper Sampling Corrector]\label{def:sampling:corrector:np}
  Fix two given properties $\property$, $\property^\prime$ of distributions on $\domain$. 
An \emph{$(\eps,\eps_1)$-non-proper sampling corrector for $\property^\prime$ assuming $\property$} is a randomized algorithm which is given parameters $\eps, \eps_1\in(0,1]$ such that $\eps_1 \geq \eps$ and $\delta\in[0,1]$, as well as sampling access to a distribution $\D$. Under the promise that $\totalvardist{\D}{\property}\leq \eps$, the algorithm must provide, with probability at least $1-\delta$ over the samples it draws and its internal randomness, sampling access to a distribution $\tilde{\D}$ such that
  \begin{enumerate}[(i)]
    \item $\tilde{\D}$ is close to $\D$: $\totalvardist{\tilde{\D}}{\D} \leq \eps_1$;
    \item $\tilde{\D}$ has the (target) property: $\tilde{\D}\in \property^\prime$.
  \end{enumerate}
  In other terms, with high probability the \new{algorithm}  will simulate exactly a sampling oracle for $\tilde{\D}$. The query complexity $q=q(\eps,\eps_1,\delta,\domain)$ of the algorithm is the number of samples from $\D$ it takes per query in the worst case.
\end{definition}
\newer{Note that if there exists $\D$ close to $\property$ such that every $\D^\prime\in\property^\prime$ is far from $\D$, this may not be achievable. Hence, the above definition requires that some relation between $\property$ and $\property^\prime$ hold: for instance, that any neighborhood of a distribution from $\property$ intersects $\property^\prime$.} Similarly, we extend this definition to non-proper improvers and batch improvers.
 
\section{A warmup: non-proper correcting of histograms}\label{sec:focus:samples:hist}
	To illustrate these ideas, we start with a toy example: non-proper correcting of \emph{regular histograms}. Recall that a distribution $\D$ over $[n]$ is said to be a \emph{$k$-histogram} if its probability mass function is piecewise-constant with at most $k$ ``pieces:'' that is, if there exists a partition $\mathcal{I}=(I_1,\dots,I_k)$ of $[n]$ into $k$ intervals such that $\D$ is constant on each $I_j$.

Letting $\mathcal{H}_{k}$ denote the class of all $k$-histograms over $[n]$, we start with the following question: given samples from a distribution close to $\mathcal{H}_{k}$, can we efficiently provide sample access to a corrected distribution $\tilde{\D}\in \mathcal{H}_{\ell}$, for some $\ell=\ell(k,\eps)$? I.e., is there a non-proper corrector for $\mathcal{H}_{\ell}$ assuming $\mathcal{H}_{k}$? \cmargin{Talk about the $k/\eps^2$ baseline to beat -- i.e., the cost of agnostic learning?}

In this short section, we show how to design such a corrector, under some additional assumption on \newer{the min-entropy of} the unknown distribution to correct. Namely, we will require the following definition: given some constant $c\geq 1$, we say that a distribution $\D$ is \emph{$c$-regular} if $\D(i) \leq \frac{c}{n}$ for all $i\in[n]$, i.e. $\norminf{\D} \leq \frac{c}{n}$.
\begin{restatable}[Correcting regular histograms]{proposition}{histogramnonpropercorrecting} \label{lemma:histogram:corrector:nonproper}
  Fix any constant $c>0$. For any $\eps$, $\eps_1 \geq 4\eps$ and $\eps_2=0$ as in the definition, there exists $\ell=O(k/\eps)$ and a non-proper sampling corrector for $\mathcal{H}_{\ell}$ assuming $\mathcal{H}_{k}$ with sample complexity $\bigO{1}$, under the assumption that the unknown distribution is $c$-regular.
\end{restatable}
\begin{proof}
  The algorithm works as follows: setting $K\eqdef \frac{ck}{\eps}$, it first divides the domain into $K\leq L\leq K+1$ intervals $I_1,\dots,I_{L}$ of size less than or equal to $\flr{\frac{n}{K}}$. Then, the corrected distribution is the ``flattening'' $\bar{\D}$ of $\D$ on these intervals: to output a sample from the $L$-histogram $\bar{\D}$, the algorithm draws a sample $s\sim\D$, checks which of the $I_i$'s this sample $s$ belongs to, and then outputs $s^\prime$ drawn uniformly from this interval. The sample complexity is clearly constant, as outputting one sample of $\bar{\D}$ only requires one from $\D$; and being an $L$-histogram, $\bar{\D}\in\mathcal{H}_{\ell}$ for $\ell \leq \frac{ck}{\eps}+1$.
  
We now turn to proving that $\totalvardist{\D}{\bar{\D}} \leq 4\eps$. Denote by $H$ the closest $k$-histogram to $\D$, i.e. $H\in\mathcal{H}_k$ such that $\alpha\eqdef\totalvardist{\D}{H}=\totalvardist{\D}{\mathcal{H}_k}$; and let $B$ be the union of the (at most $k$) intervals among $I_1,\dots,I_{L}$ where $H$ is not constant. Since $\D$ is $c$-regular, we do have $\D(B) \leq k\cdot\frac{c}{n}\cdot\frac{n}{K} = \eps$. Then, since $H$ is $\alpha$-close to $\D$ we get $H(B) \leq \eps+\alpha$.

Now, let $\bar{\D}$ (resp. $\bar{H}$) be the $L$-histogram obtained by ``flattening'' $\D$ (resp. $H$) on $I_1,\dots,I_{L}$. By the data processing inequality (\autoref{lemma:data:processing:inequality:total:variation}), we obtain
\[
    \totalvardist{\bar{\D}}{\bar{H}} \leq \totalvardist{\D}{H}.
\]
Therefore, by the triangle inequality,
\begin{align*}
    \totalvardist{\D}{\bar{\D}} &\leq \totalvardist{\D}{H} + \totalvardist{H}{\bar{H}} + \totalvardist{\bar{H}}{\bar{\D}} 
    \leq 2\totalvardist{\D}{H} + \totalvardist{H}{\bar{H}}.
\end{align*}
Furthermore, as $H$ and $\bar{H}$ can only differ on $B$, and since the flattening operation preserve the probability weight on each interval of $\mathcal{I}$, we obtain
\begin{align*}
  \totalvardist{H}{\bar{H}} &= \frac{1}{2}\normone{ H-\bar{H} } = \frac{1}{2}\sum_{i\in B} \abs{ H(i) - \bar{H}(i) } 
  \leq \frac{1}{2}\left(H(B) + \bar{H}(B)\right) =H(B) \leq \eps+\alpha
\end{align*}
which, once plugged back in the previous expression, yields
\begin{align*}
    \totalvardist{\D}{\bar{\D}} &\leq 2\totalvardist{\D}{H} + \eps+\alpha = 3\alpha + \eps 
    \leq 4\eps
\end{align*}
since $\alpha \leq \eps$ by assumption.
\end{proof}
 
\section{Connections to learning and testing}\label{sec:connections}
	In this section, we draw connections between sampling improvers and other areas, namely testing and learning. These connections shed light on the relation between our model and these other lines of work, and provide a way to derive new algorithms and impossibility results for both testing or learning problems. (For the formal definition of the testing and learning notions used in this section, the reader is referred to~\autoref{appendix:definitions} and the relevant subsections.)

\subsection{From learning to correcting}\label{ssec:connection:learning:approach}As a first observation, it is not difficult to see that, under the assumption that the unknown distribution $\D$ belongs to some specific class \class, correcting (or improving) a property $\property$ requires at most as many samples as learning the class \class; that is, \emph{learning (a class of distributions) is at least as hard as correcting (distributions of this class).} Here, \property and \class need not be related.

Indeed, assuming there exists a learning algorithm \Learner for \class, it then suffices to run $\Learner$ on the unknown distribution $\D\in\class$ to learn (with high probability) a hypothesis $\hat{\D}$ such that $\D$ and $\hat{\D}$ are at most at distance $\frac{\eps_1-\eps}{2}$. In particular, $\hat{\D}$ is at most $\frac{\eps_1+\eps}{2}$-far from \property. One can then (e.g., by exhaustive search) find a distribution $\tilde{\D}$ in $\property$ which is closest to $\hat{\D}$ (and therefore at most $\eps_1$-far from \D), and use it to produce as many ``corrected samples'' as wanted:

\begin{theorem}\label{theo:connection:regular:learning}
Let $\class$ a class of probability distributions over $\domain$. Suppose there exists a learning algorithm $\Learner$ for $\class$ with sample complexity $q_\Learner$.
Then, for any property $\property$ of distributions, there exists a (not-necessarily computationally efficient) sampling corrector for \property with sample complexity $q(\eps,\eps_1,\delta)=q_\Learner\left(\frac{\eps_1-\eps}{2},\delta\right)$, under the promise that $\D\in\class$.
\end{theorem}

Furthermore, if the (efficient) learning algorithm \Learner has the additional guarantee that its hypothesis class is a subset of $\property$ (i.e., the hypotheses it produces always belong to \property) and that the hypotheses it contains allow efficient generation of samples, then we immediately obtain a computationally efficient sampling corrector: indeed, in this case $\hat{\D}\in\property$ already.
Furthermore, as mentioned in the introduction, when efficient agnostic proper learning algorithms for distribution classes exist, then there are efficient sampling correctors for the same classes.  It is however worth pointing out that this correcting-by-learning approach is quite inefficient with regard to the amount of extra randomness needed: indeed, every sample generated from $\tilde{\D}$ requires fresh new random bits.

To illustrate this theorem, we give two easy corollaries. The first follows from Chan et al., who showed in~\cite{CDSS:13} that monotone hazard risk distributions can be learned to accuracy $\eps$ using $\tildeO{{\log n}/{\eps^4}}$ samples; moreover, the hypothesis obtained is a $\bigO{\log (n/\eps)/\eps^2}$-histogram.
\begin{corollary}
Let $\class$ be the class of monotone hazard risk distributions over $[n]$, and $\property$ be the property of being
a histogram with (at most) $\sqrt{n}$ pieces. Then, under the promise that 
$\D\in\class$ and as long as $\eps=\tilde{\Omega}(1/\sqrt{n})$, there is a sampling corrector for $\property$ with sample complexity $\tildeO{\frac{\log n}{(\eps_1-\eps)^4}}$.
\end{corollary}
\noindent Our next example however demonstrates that this learning approach is not always optimal:
\begin{corollary}
Let $\class$ be the class of monotone distributions over $[n]$, and $\property$ be the property of being
a histogram with (at most) $\sqrt{n}$ pieces. Then, under the promise that 
$\D\in\class$ and as long as $\eps=\tilde{\Omega}(1/\sqrt{n})$, there is a sampling corrector for $\property$ with sample complexity $\bigO{\frac{\log n}{(\eps_1-\eps)^3}}$.
\end{corollary}
\noindent Indeed, for learning monotone distributions $\bigTheta{{\log n}/{\eps^3}}$ samples are known to be necessary and sufficient~\cite{Birge:87}. 
Yet, one can also correct the distribution by simulating samples directly from its flattening on the corresponding Birg\'e decomposition (as per~\autoref{def:birge:decomposition}); and every sample from this correction-by-simulation costs exactly \emph{one} sample from the original distribution.

\subsection{From correcting to agnostic learning}\label{ssec:connections:agnostic}

Let $\class$  and $\mathcal{H}$ be two classes of probability distributions over $\domain$. Recall that a \emph{(semi-)agnostic learner for \class} (using hypothesis class $\mathcal{H}$) is a learning algorithm \Algo which, given sample access to an arbitrary distribution $\D$ 
and parameter $\eps$, outputs a hypothesis $\hat{\D}\in\mathcal{H}$ such that, with high probability, $\hat{\D}$ does ``as much as well as the best approximation from $ \class$:''
\[
	\totalvardist{\D}{\hat{\D}} \leq c\cdot\opt_{\class,\D} + \bigO{\eps}
\]
where $\opt_{\class,\D}\eqdef \inf_{\D_{\class}\in\class} \totalvardist{\D_{\class}}{\D}$ and $c\geq 1$ is some absolute constant (if $c=1$, the learner is said to be agnostic).

\paragraph*{} We first describe how to combine a (non-agnostic) learning algorithm with a sampling corrector in order to obtain an agnostic learner, under the strong assumption that a (rough) estimate of $\opt$ is known. Then, we explain how to get rid of this extra requirement, using machinery from the distribution learning literature (namely, an efficient \emph{hypothesis selection} procedure).

\begin{restatable}{theorem}{theoconnectionagnostic}\label{theo:connection:agnostic:learning}
Let $\class$ be as above. Suppose there exists a learning algorithm $\Learner$ for $\class$ with sample complexity $q_\Learner$, and a batch sampling corrector $\Algo$ for $\class$ with sample complexity $q_\Algo$. Suppose further that a constant-factor estimate $\widehat{\opt}$ of $\opt_{\class, \D}$ is known (up to a multiplicative $c$).

Then, there exists a \new{semi-}agnostic learner for \class with sample complexity $q(\eps,\delta)=q_\Algo(\widehat{\opt},\widehat{\opt}+\eps, q_\Learner(\eps,\frac{\delta}{2}),\frac{\delta}{2})$ (where the constant in front of $\opt_{\class, \D}$ is $c$).
\end{restatable}
\begin{proof}
Let $c$ be the constant such $\opt_{\class,\D} \leq \widehat{\opt}\leq c\cdot\opt_{\class,\D}$. The agnostic learner $\Learner^\prime$ for \property, on input $\eps\in(0,1]$, works as follows:
\begin{itemize}[-]
  \item\label{theo:connection:agnostic:learning:step:1} Run \Algo on $\D$ with parameters $(\widehat{\opt}, \widehat{\opt}+\eps,\frac{\delta}{2})$ to get $q_\Learner(\eps,\frac{\delta}{2})$ samples distributed according to some distribution $\tilde{\D}$.
  \item Run \Learner on these samples, with parameters $\eps,\frac{\delta}{2}$, and output its hypothesis $\hat{D}$.
\end{itemize}
We hereafter condition on both algorithms succeeding (which, by a union bound, happens with probability at least $1-\delta$). Since $\D$  is $\widehat{\opt}$-close to $\class$, and therefore by correctness of the sampling corrector we have both $\tilde{\D}\in\class$ and $\totalvardist{\D}{\tilde{\D}} \leq \widehat{\opt}+\eps$. Hence, the output $\hat{\D}$ of the learning algorithm satisfies $\totalvardist{\tilde{\D}}{\hat{\D}} \leq \eps$, which implies
\begin{equation}
	\totalvardist{\D}{\hat{\D}} \leq \widehat{\opt} + 2\eps \leq c\cdot\opt_{\class,\D} + 2\eps \label{eq:connection:agnostic:learning}
\end{equation}
for some absolute constant $c$, as claimed (using the assumption on $\widehat{\opt}$).
\end{proof}

It is worth noting that in the case the learning algorithm is \emph{proper} (meaning the hypotheses it outputs belong to the target class \class: that is, $\mathcal{H}\subseteq\class$), then so is the agnostic learner obtained with \autoref{theo:connection:agnostic:learning}. This turns out to be a very strong guarantee: specifically, getting (computationally efficient) proper agnostic learning algorithms remains a challenge for many classes of interest -- see e.g. \cite{DDS:PBD:12}, which mentions efficient proper learning of Poisson Binomial Distributions as an open problem.

We stress that the above can be viewed as a \emph{generic} framework to obtain efficient agnostic learning results from known efficient learning algorithms.
For the sake of illustration, let us consider the simple case of Binomial distributions: it is known, for instance as a consequence of the aforementioned results on PBDs, that learning such distributions can be performed with $\tildeO{1/\eps^2}$ samples (and that $\bigOmega{1/\eps^2}$ are required). Our theorem then provides a simple way to obtain agnostic learning of Binomial distributions with sample complexity $\tildeO{1/\eps^2}$: namely, by designing  an efficient sampling corrector for this class with sample complexity $\poly(\log\frac{1}{\eps},\log\frac{1}{\eps_1})$.

\begin{corollary}
Suppose there exists a batch sampling corrector $\Algo$ for the class $\mathcal{B}$ of Binomial distributions over $[n]$, with sample complexity $q_\Algo(\eps,\eps_1,m,\delta)=\poly\!\log(\frac{1}{\eps},\frac{1}{\eps_1},m,\frac{1}{\delta})$.
Then, there exists a semi-agnostic learner for $\mathcal{B}$, which, given access to an unknown distribution $\D$ promised to be $\eps$-close to some Binomial distribution, takes $\tildeO{\frac{1}{\eps^2}}$ samples from $\D$ and outputs a distribution $\hat{B}\in\mathcal{B}$ such that
\[
  \totalvardist{\D}{\hat{B}} \leq 3\eps
\]
with probability at least $2/3$.
\end{corollary}
\noindent To the best of our knowledge, an agnostic learning algorithm for the class of Binomial distributions with sample complexity $\tildeO{1/\eps^2}$ is not explicitly known, although the results of \cite{CDSS:14} do imply a $\tildeO{1/\eps^3}$ upper bound and a modification of~\cite{DDS:PBD:12} (to make their algorithm agnostic) seems to yield one. The above suggests an approach which would lead to the (essentially optimal) sample complexity. (Since publication of our work, we have learned that~\cite{ADLS:15} provides such a result unconditionally.)

\subsubsection{Removing the assumption on knowing \texorpdfstring{$\widehat{\opt}$}{an estimate of \opt}}

In the absence of such an estimate $\widehat{\opt}$ within a constant factor of $\opt_{\class, \D}$ given as input, one can apply the following strategy, inspired of \cite[Theorem 6]{CDSS:14:NIPS}. In the first stage, we try to repeatly ``guess'' a good $\widehat{\opt}$, and run the agnostic learner of \autoref{theo:connection:agnostic:learning} with this value to obtain a hypothesis. After this stage, we have generated a succinct list $\mathcal{H}$ of hypotheses, one for each $\widehat{\opt}$ that we tried: the second stage is then to run a hypothesis selection procedure to pick the best $h\in\mathcal{H}$: as long as one of the guesses was good, this $h$ will be an accurate hypothesis.

More precisely, suppose we run the agnostic learner of \autoref{theo:connection:agnostic:learning} a total of $\log(1/\eps)$ times, setting at the $k$\textsuperscript{th} iteration $\widehat{\opt}_k\eqdef 2^k\eps$ and $\delta^\prime\eqdef \delta/(2\log(1/\eps))$. For the first $k$ such that $2^{k-1}\eps \leq \opt_{\class,\D} < 2^k\eps$, $\widehat{\opt}_k$ is in $[\opt_{\class,\D}, 2\cdot \opt_{\class,\D}]$. Therefore, by a union bound on all runs of the learner at least one of the hypotheses $\hat{\D}_k$ will have the agnostic learning guarantee we want to achieve; i.e. will satisfy \eqref{eq:connection:agnostic:learning}, with $c=2$.

Conditioned on this being the case, it remains to determine \emph{which} hypothesis achieves the guarantee of being $(2\opt+\bigO{\eps})$-close to the distribution $\D$. This is where 
we apply a hypothesis selection algorithm~--~a ``tournament'' procedure, as in e.g.~\cite{DL:01,DDS:PBD:12,DK:13,AJOS:14:ISIT}~--~to our $N=\log(1/\eps)$ candidates, 
with accuracy parameter $\eps$ and failure probability $\delta/2$. This algorithm has the following guarantee:
\begin{proposition}[{\cite[Chapter 7]{DL:01}}]\label{prop:tournament}
  There exists a procedure \textsc{Tournament} that, given sample access to an unknown distribution $\D$ and both sample and evaluation access to $N$ hypotheses $H_1,\dots,H_N$, has the following behavior. \textsc{Tournament} makes a total of $\bigO{\log N \log(1/\delta)/{\eps^2}}$ queries to $\D,H_1,\dots,H_N$, runs in time $\bigO{N^2}$, and outputs a hypothesis $H_i$ such that, with probability at least $1-\delta$,
  \[
      \totalvardist{\D}{H_i} \leq 3 \min_{j\in[N] } \totalvardist{\D}{H_j} + \bigO{\eps}.
  \]
\end{proposition}
We note that the quadratic running time can be brought down to near-linear, as shown in~\cite{DK:13,AJOS:14:ISIT} (with the same sample complexity, although at the price of a worse semi-agnostic constant). This improved running time, however, is not crucial for our applications.

\paragraph*{Summary.} Using this result in the approach outlined above, we get with probability at least $1-\delta$, we will obtain a hypothesis $\hat{\D}_{k^\ast}$ doing ``almost as well as the best $\D_k$''; that is,
\[
	\totalvardist{ \D }{ \hat{\D}_{k^\ast} } \leq 6\cdot \opt_{\class,\D} + \bigO{\eps}
\]
The overall sample complexity is
\[
	\sum_{k=1}^{\log(1/\eps)}q_\Algo\!\left(2^k\eps,(2^{k}+1)\eps, q_\Learner\!\left(\eps,\frac{\delta}{4\log(1/\eps)}\right),\frac{\delta}{4\log(1/\eps)}\right) + \tildeO{ \frac{1}{\eps^2}\log\frac{1}{\delta} } 
\]
where the first term comes from the $\log(1/\eps)$ runs of the learner from \autoref{theo:connection:agnostic:learning}, and the second is the overhead due to the hypothesis selection tournament.

\subsection{From correcting to tolerant testing}

We observe that the existence of sampling 
correctors for a given property $\property$, 
along with an efficient distance estimation procedure, 
allows one to convert any distribution testing algorithm into a tolerant 
distribution testing algorithm.
This is similar to the connection between ``local reconstructors'' and 
tolerant testing of graphs described in \cite[Theorem 3.1]{Brakerski:08}
and \cite[Theorem 3.1]{CGR:12}. That is, 
if a property \property has both a distance estimator and a sampling corrector, 
then one can perform \emph{tolerant} testing of \property in the time 
required to generate enough corrected samples for both the estimator and a 
(non-tolerant) tester. \medskip

We first state our theorem in all generality, 
before instantiating it in several corollaries. For the sake of clarity,
the reader may wish to focus on these on a first pass.

\begin{restatable}{theorem}{theoconnectiontolerantbrakersky}\label{theo:connection:corrector:testing}
Let \class be a class of distributions, and $\property\subseteq\class$ a property. Suppose there exists an $(\eps,\eps_1)$-batch sampling corrector \Algo for \property with complexity $q_\Algo$, and a distance estimator $\mathcal{E}$ for \class with complexity $q_{\mathcal{E}}$ -- that is, given sample access to $\D_1,\D_2\in\class$ and parameters \eps, $\delta$, $\mathcal{E}$ draws $q_{\mathcal{E}}(\eps,\delta)$ samples from $\D_1,\D_2$ and outputs a value $\hat{d}$ such that $\abs{ \hat{d}-\totalvardist{\D_1}{\D_2} } \leq \eps$ with probability at least $1-\delta$.

Then, from any property tester $\Tester$ for \property with sample complexity $q_\Tester$, one can get a tolerant tester $\Tester^\prime$ with query complexity $q(\eps^\prime,\eps,\delta)=q_\Algo\!\left(\eps^\prime, \bigTheta{\eps}, q_{\mathcal{E}}(\frac{\eps-\eps^\prime}{4},\frac{\delta}{3})+q_\Tester(\frac{\eps-\eps^\prime}{4},\frac{\delta}{3}), \frac{\delta}{3}\right)$.
\end{restatable}
\ifnum\fulldetails=1 \begin{proof}
The tolerant tester $\Tester^\prime$ for \property, on input $0 \leq \eps^\prime < \eps \leq 1$, works as follows, setting $\beta\eqdef\frac{\eps-\eps^\prime}{4}$ and $\eps_1\eqdef \eps^\prime+\beta$:
\begin{itemize}[-]
  \item\label{theo:connection:corrector:testing:step:1} Run \Algo on $\D$ with parameters $(\eps^\prime,\eps_1,\delta/3)$ to get $q_{\mathcal{E}}(\beta,\delta/3)+q_\Tester(\beta,\delta/3)$ samples distributed according to some distribution $\tilde{\D}$. Using these samples:
    \begin{enumerate}
        \item\label{theo:connection:corrector:testing:step:2} Estimate $\totalvardist{\D}{\tilde{D}}$ to within an additive $\beta$, and \reject if this estimate is more than $\eps_1+\beta=\frac{\eps+\eps^\prime}{2}$;
        \item\label{theo:connection:corrector:testing:step:3} Otherwise, run \Tester on $\tilde{D}$ with parameter $\beta$ and accept if and only if \Tester outputs \accept.
    \end{enumerate}
\end{itemize}

We hereafter condition on all 3 algorithms succeeding (which, by a union bound, happens with probability at least $1-\delta$).

\noindent If $\D$ is $\eps^\prime$-close to \property, then the corrector ensures that $\tilde{\D}$ is $\eps_1$-close to $\D$, so the estimate of $\totalvardist{\D}{\tilde{\D}}$ is at most $\eps_1+\beta$: Step~\ref{theo:connection:corrector:testing:step:2} thus passes, and as $\tilde{\D}\in\property$ the tester outputs \accept in Step~\ref{theo:connection:corrector:testing:step:3}.\medskip

\noindent On the other hand, if $\D$ is $\eps$-far from \property, then either \textsf{(a)} $\totalvardist{\D}{\tilde{\D}} > \eps_1+2\beta$ (in which case we output \reject in Step~\ref{theo:connection:corrector:testing:step:2}, since the estimate exceeds $\eps_1+\beta$), or \textsf{(b)} $\totalvardist{\tilde{\D}}{\property} > \eps - (\eps_1+2\beta) = \beta$, in which case $\Tester$ outputs \reject in Step~\ref{theo:connection:corrector:testing:step:3}.
\end{proof}
\fi 
\begin{remark}\label{remark:distance:estimation}
Only asking that the distance estimation procedure $\mathcal{E}$ be specific to the class \class is not innocent; indeed, it is known (\cite{ValiantValiant:11}) that for \emph{general} distributions, distance estimation has sample complexity $n^{1-o(1)}$. However, the task becomes significantly easier for certain classes of distributions: and for instance can be performed with only $\tildeO{k\log n}$ samples, if the distributions are guaranteed to be $k$-modal \cite{DDSVV:13}. This observation can be leveraged in cases
when one knows that the distribution has a specific property, but does not quite satisfy a second property: e.g. is known to be $k$-modal but not known to be, say, log-concave.
\end{remark}

The reduction above can be useful both as a black-box 
way to derive upper bounds for tolerant testing, 
as well as to prove lower bounds for either testing or distance estimation. 
For the first use, we give two applications of 
our theorem to provide tolerant monotonicity testers for $k$-modal distributions. 
The first is a conditional result, 
showing that the existence of
good monotonicity correctors yield tolerant testers. The second, 
while unconditional, only guarantees a weaker form of tolerance
(guaranteeing acceptance only of distributions that are very close to monotone); 
and relies on a corrector we describe in~\autoref{sec:monotonicity:oblivious}. 
As we detail shortly after stating these two results, 
even this weak tolerance improves upon the one provided by currently 
known testing algorithms.

\begin{corollary}\label{coro:connection:corrector:testing}
Suppose there exists an $(\eps,\eps_1)$-batch sampling corrector for monotonicity with complexity $q$. Then, for any $k=\bigO{\log n/\log\log n}$, there exists an algorithm that distinguishes whether a $k$-modal distribution is  \textsf{(a)} $\eps$-close to monotone or  \textsf{(b)} $5\eps$-far from monotone with success probability $2/3$, and sample complexity
\[
	q\!\left( \eps, 2\eps, C\frac{k\log n}{\eps^4\log\log n}, \frac{1}{9}  \right)
\]
where $C$ is an absolute constant.
\end{corollary}
\begin{proof}
 We combine the distance estimator of \cite{DDSVV:13} with the monotonicity tester of \cite[Section 3.4]{DDS:12}, which both apply to  the class of $k$-modal distributions. As their respective sample complexity is, for distance parameter $\alpha$ and failure probability $\delta$, $\bigO{\left(\frac{k^2}{\alpha^4}+\frac{k\log n}{\alpha^4\log(k\log n)}\right)\log\frac{1}{\delta}}$ and $\bigO{\frac{k}{\alpha^2}\log\frac{1}{\delta}}$, the choice of parameters ($\delta=1/3$, $\eps$ and $5\eps$) and the assumption on $k$ yield
 \[
   \bigO{\frac{k}{\eps^2}} + \bigO{\frac{k^2}{\eps^4}+\frac{k\log n}{\eps^4\log(k\log n)} }
   = \bigO{\frac{k\log n}{\eps^4\log(k\log n)}}
 \]
  and we obtain by \autoref{theo:connection:corrector:testing} a tolerant tester with sample complexity $q\!\left( \eps, 2\eps, \bigO{\frac{k\log n}{\eps^4\log(k\log n)}}, \frac{1}{9}  \right)$, as claimed.
\end{proof}

Another application of this theorem, but this time taking advantage of a result from \autoref{sec:monotonicity}, allows us to derive an \emph{explicit} tolerant tester for monotonicity of $k$-modal distributions:
\begin{corollary}\label{coro:connection:corrector:testing:oblivious}
For any $k\geq 1$, there exists an algorithm that distinguishes whether a $k$-modal distribution is  \textsf{(a)} $\bigO{\eps^3/\log^2 n}$-close to monotone or  \textsf{(b)} $\eps$-far from monotone with success probability $2/3$, and sample complexity
\[
	\bigO{\frac{1}{\eps^4}\frac{k\log n}{\log (k\log n)} +\frac{k^2}{\eps^4} }.
\]
In particular, for $k=\bigO{\log n/\log\log n}$ this yields a (weakly) tolerant tester with sample complexity $\bigO{\frac{1}{\eps^4}\frac{k\log n}{\log\log n} }$.
\end{corollary}
\begin{proof}
 We again use the distance estimator of~\cite{DDSVV:13} and the monotonicity tester of~\cite{DDS:12}, which both apply to  the class of $k$-modal distributions, this time with the monotonicity corrector we describe in \autoref{lemma:oblivious:corrector:monotonicity:weak}, which works for any $\eps_1$ and $\eps=\bigO{\eps_1^3/\log^2 n}$ and has constant-rate sample complexity (that is, it takes $\bigO{q}$ samples from the original distribution to output $q$ samples). Similarly to \autoref{coro:connection:corrector:testing}, the sample complexity is a straightforward application of \autoref{theo:connection:corrector:testing}.
\end{proof}

Note that, to the best of our knowledge, no tolerant tester for 
monotonicity of $k$-modal distributions was previously known, 
though using the (regular) $\bigO{k/\eps^2}$-sample tester 
of~\cite{DDS:12} and standard arguments, one can achieve
a weak tolerance on the order of $\bigO{\eps^2/k}$. 
While the sample complexity obtained
in~\autoref{coro:connection:corrector:testing:oblivious} is 
worse by a $\poly\!\log(n)$ factor, 
it has better tolerance for $k=\Omega(\log^2 n/\eps)$.
 
\section{Sample complexity of correcting monotonicity}\label{sec:focus:samples}
	In this section, we focus on the sample complexity aspect of correcting, considering the specific example of monotonicity correction. As a first result, we show in~\autoref{sec:monotonicity} how to design a simple batch corrector for monotonicity which, after a preprocessing step costing logarithmically many samples, is able to answer an arbitrary number of queries. This corrector follows the ``learning approach'' described in~\autoref{ssec:connection:learning:approach}, and in particular provides a very efficient way to amortize the cost of making many queries to a corrected distribution.

A natural question is then whether one can ``beat'' this approach, and correct the distribution without approximating it as a whole beforehand. \autoref{sec:monotonicity:oblivious} answers it by the affirmative: namely, we show that one can correct distributions that are guaranteed to be $(1/\log^2 n)$-close to monotone in a completely \emph{oblivious} fashion, with a non-adaptive approach that does not require to learn anything about the distribution.

Finally, we give in \autoref{sec:monotonicity:cdf} a corrector for monotonicity with no restriction on the range of parameters, but assuming a stronger type of query access to the original distribution. \newerest{Specifically}, our algorithm leverages the ability to make \emph{\cdf queries} to the distribution $\D$, in order to generate independent samples from a corrected $\tilde{\D}$. This sampling corrector also outperforms the one from \autoref{sec:monotonicity}, making only $\bigO{\sqrt{\log n}}$ queries per sample on expectation.

\new{\paragraph{A parenthesis: non-proper correcting.} We note that it is easy to obtain a non-proper corrector for $k(n,\eps)$-histograms assuming monotonicity with \emph{constant} sample complexity, for $k(n,\eps) = \bigTheta{\frac{\log n}{\eps}}$. Indeed, this follows from the oblivious Birg\'e decomposition \newer{(see~\autoref{def:birge:decomposition})} we shall be using many times through this section, which ensures that ``flattening'' a monotone distribution yields a $k(n,\eps)$-histogram that remains close to the original distribution.}

\subsection{A natural approach: correcting by learning}\label{sec:monotonicity}
Our first corrector works in a straightforward fashion: 
it \emph{learns} a good approximation of the distribution to correct,
which is also concisely represented.   
It then uses this approximation to 
build a sufficiently good monotone distribution $M^\prime$ ``offline,''
by searching for the closest monotone distribution, which in
this case can be achieved via linear programming.
Any query made to the corrector is then answered according to the latter distribution, at no additional cost.
\begin{restatable}[Correcting by learning]{lemma}{monotonicitylearningcorrector} \label{lemma:learning:corrector:birge}
  Fix any constant $c>0$. For any $\eps$, $\eps_1 \geq (3+c)\eps$ and $\eps_2=0$ as in the definition, any type of oracle \ORACLE and any number of queries $m$, there exists a sampling corrector for monotonicity from \new{sampling} to \ORACLE with sample complexity $\bigO{\log n/\eps^3}$.
\end{restatable}
\begin{proof}  Consider the Birg\'e decomposition $\mathcal{I}_\alpha=(I_1,\dots,I_\ell)$ with parameter $\alpha\eqdef\frac{c \eps}{3}$ which partitions the domain $[n]$ into $\bigO{\frac{\log n}{\eps}}$ intervals.
  By \autoref{coro:Birge:decomposition:robust} and the learning result of \cite{Birge:87}, we can learn with $\bigO{\frac{\log n}{\eps^3}}$ samples a $\bigO{\frac{\log n}{\eps}}$-histogram $\bar{\D}$ such that:
   \begin{equation}
   \totalvardist{\D}{\bar{\D}}\leq 2\eps+\alpha.
\end{equation}    
  Also, let $M$ be the closest monotone distribution to $\D$. From Eq.~\eqref{eq:Birge:tv}, we get the following: letting $\mathcal{M}$ denote the set of monotone distributions,   
 \begin{equation}
 \totalvardist{\bar{\D}}{\mathcal{M}}=\totalvardist{\Phi_{\alpha}(\D)}{\mathcal{M}}\leq \totalvardist{\Phi_{\alpha}(\D)}{\Phi_{\alpha}(M)}\leq \totalvardist{\D}{M}\leq \eps 
 \end{equation}
  where the first inequality follows from the fact that $\Phi_{\eps}(M)$ is monotone. Thus, $\bar{\D}$ is $\eps$-close to monotone, which implies that $\bar{\D}^\prime$ is $(\eps+\alpha)$-close to monotone. Furthermore, it is easy to see that, without loss of generality, one can assume the closest monotone distribution $\bar{\D}^\prime$ to be piecewise constant with relation to the same partition (e.g., using again Eq.~\eqref{eq:Birge:tv}). It is therefore sufficient to find such a piecewise constant distribution: to do so, consider the following linear program which finds exactly this: a monotone $M^\prime$, closest to $\bar{\D}^\prime$ and piecewise constant on $\mathcal{I}_\alpha$:
\begin{alignat*}{2}
    \text{minimize }   & \sum_{j=1}^\ell \abs{ x_j - \frac{\bar{\D}^\prime(I_j)}{\abs{ I_j }} } \cdot \abs{ I_j }    \\
    \text{subject to } & 1\geq x_1\geq x_2 \geq \dots \geq x_l \geq 0 & \\
                       &  \sum_{j=1}^\ell x_j\abs{ I_j } =1
  \end{alignat*}  
\noindent This linear program has $\bigO{\frac{\log n}{\eps}}$ variables and so it can be solved in time $\poly(\log n,\frac{1}{\eps})$ .  

After finding a solution $(x_j)_{j\in[\ell]}$ to this linear program,\footnote{To see why a good solution always exists, consider the closest monotone distribution to $\bar{\D}$, and apply $\Phi_{\alpha}$ to it. This distribution satisfies all the constraints.}  we define the distribution $M^\prime\colon[n]\to [0,1]$ as follows: $M^\prime(i)=x_{\operatorname{ind}(i)}$, where $\operatorname{ind}(i)$ is the index of the interval of $\mathcal{I}_\alpha$ which $i$ belongs to. 
This implies that
\begin{equation*}
\totalvardist{\bar{\D}^\prime}{M^\prime}\leq \eps+\alpha
\end{equation*}
and by the triangle inequality we finally get:
\[ \totalvardist{\D}{M^\ast}\leq \totalvardist{\D}{\bar{\D}}+\totalvardist{\bar{\D}}{\bar{\D}^\prime}+\totalvardist{\bar{\D}^\prime}{M^\prime}\leq 3\eps+3\alpha = (3+c)\eps. \]
\end{proof}

\subsection{Oblivious correcting of distributions which are very close to monotone}\label{sec:monotonicity:oblivious}

We now turn to our second monotonicity corrector, which achieves constant sample complexity for distributions already $(1/\log^2 n)$-close to monotone. \newer{Note that this is a very strong assumption, as if one draws less than $\log^2 n$ samples one does not expect to see any difference between such a distribution $\D$ and its closest monotone distribution. Still, our construction actually yields a stronger guarantee: namely, given \emph{evaluation (query)} access to $\D$, it can answer evaluation queries to the corrected distribution as well.} See~\autoref{rk:oblivious:stronger:queries} for a more detailed statement.

The high-level idea is to treat the distribution as a $k$-histogram on the Birg\'e decomposition (for $k=\bigO{\log n}$), thus ``implicitly approximating'' it; and to correct this histogram by adding a certain amount of probability weight to every interval, so that each gets slightly more than the next one. By choosing these quantities carefully, this ensures that \emph{any} violation of monotonicity gets corrected in the process, without ever having to find out \emph{where} they actually occur.\medskip

We start by stating the general correcting approach for general $k$-histograms satisfying a certain property (namely, the ratio between two consecutive intervals is constant).
\begin{restatable}{lemma}{monotonicityobliviouscorrector}\label{theorem:oblivious:corrector}
Let $\mathcal{I}=(I_1,\dots,I_{k})$ be a decomposition of $[n]$ in consecutive intervals such that $\abs{I_{j+1}}/\abs{I_j} = 1+c$ for all $j$, and $\D$ be a $k$-histogram distribution on $\mathcal{I}$ that is $\eps$-close to monotone. Then, there is a monotone distribution $\tilde{\D}$ which can be sampled from in constant time given oracle access to $\D$, such that $\totalvardist{\D}{\tilde{\D}} = \bigO{\eps k^2}$.
Further, $\tilde{\D}$ is also a $k$-histogram distribution on $\mathcal{I}$.
\end{restatable}
\ifnum\fulldetails=1     \begin{proof} We will argue that no interval can have significantly more total weight than the previous one, as it would otherwise contradict the bound on the closeness to monotonicity. This bound on the ``jump'' between two consecutive intervals enables us to define a new distribution $\hat{\D}$ which is a mixture of $\D$ with an arithmetically decreasing $k$-histogram (which only depends on $\eps$ and $k$); it can be shown that for the proper choice of parameters, $\hat{\D}$ is now monotone.\smallskip

    We start with the following claim, which leverages the distance to monotonicity in order to give a bound on the total violation between two consecutive intervals of the partition:
    \begin{claim}\label{claim:oblivious:monotone:guarantee:buckets}
      Let $\D$ be a $k$-histogram distribution on $\mathcal{I}$ that is $\eps$-close to monotone. Then, for any $j\in\{1,\dots,k-1\}$,
      \begin{equation}\label{eq:oblivious:monotone:guarantee:buckets}
	      D(I_{j+1}) \leq (1+c)D(I_{j})+\eps(2+c).
      \end{equation}
    \end{claim}
    \begin{proof}
    First, observe that without loss of generality, one can assume the monotone distribution closest to $\D$ to be a $k$-histogram on $\mathcal{I}$ as well (e.g., by a direct application of \autoref{lemma:data:processing:inequality:total:variation} to the flattening on $\mathcal{I}$ of the monotone distribution closest to $\D$). Assume there exists an index $j\in\{1,\dots,k-1\}$ contradicting \eqref{eq:oblivious:monotone:guarantee:buckets}; then,
    \[
	    \frac{\D(I_{j+1})}{\abs{I_{j+1}}} > (1+c)\frac{\D(I_{j})}{\abs{I_{j}}}\cdot\frac{\abs{I_{j}}}{\abs{I_{j+1}}} + \eps\frac{ 2+c }{ \abs{I_{j+1}}  } = \frac{\D(I_{j})}{\abs{I_{j}}} +   \eps\frac{ 2+c }{ \abs{I_{j+1}}  }.
    \]
    But any  monotone distribution $M$ which is a $k$-histogram on $\mathcal{I}$ must satisfy $\frac{M(I_{j+1})}{\abs{I_{j+1}}} \leq \frac{M(I_{j})}{\abs{I_{j}}}$; so that at least  $\eps(2+c)$ total \new{weight} has to be ``redistributed'' to fix this violation.     Indeed, it is not hard to see\footnote{E.g., by writing the $\lp[1]$ cost as the sum of the weight added/removed from ``outside'' the two buckets and the \new{weight} moved between the two buckets in order to satisfy the monotonicity condition, and minimizing this function.}
     that the minimum amount of probability weight to ``move'' in order to do so is at least what is needed to uniformize $\D$ on $I_j$ and $I_{j+1}$. 
    This latter process yields a distribution $\D^\prime$  
    which puts weight $(\D(I_j)+\D(I_{j+1}))/((2+c)\abs{I_j})$
    on each element of $I_j\cup I_{j+1}$, and the total variation distance between $\D$ and $\D^\prime$ (a lower bound on its distance to monotonicity) is then
    \[
	    \totalvardist{\D}{\D^\prime}=\frac{\D(I_{j+1})+\D(I_j)}{2+c}-\D(I_j) = \frac{\D(I_{j+1})-(1+c)\D(I_j)}{2+c} > \frac{\eps(2+c)}{2+c}=\eps
    \]
    which is a contradiction.
    \end{proof}
    This suggests immediately the following correcting scheme: to output samples according to $\tilde{\D}$, $k$-histogram on $\mathcal{I}$ defined by
    \begin{align*}
      \tilde{\D}(I_k) &= \lambda\left( \D(I_k)  \right) \\ 
      \tilde{\D}(I_{k-1}) &= \lambda\left( \D(I_{k-1}) + (2+c)\eps  \right) \\ 
      &\vdots\\
      \tilde{\D}(I_{k-j}) &= \lambda\left( \D(I_{k-j}) + j(2+c)\eps  \right) \\ 
    \end{align*}
    that is
    \begin{align*}
      \tilde{\D}(I_j) &= \lambda\left( \D(I_j) + \eps\sum_{i=j}^{k-1}\left(1+\frac{\abs{I_{j+1}}}{\abs{I_j}}\right)  \right) \qquad 1\leq j\leq k
    \end{align*}
    where the normalizing factor is $\lambda\eqdef \left( 1+\eps(2+c)\frac{k(k-1)}{2} \right)^{-1}$. As, 
    by \autoref{claim:oblivious:monotone:guarantee:buckets}, 
    adding weight decreasing by $(2+c)\eps$ at each step fixes any pair of adjacent intervals
    whose average weights are not monotone,
    $\tilde{\D}/\lambda$ is a non-increasing non-negative function. The normalization by $\lambda$ preserving the monotonicity, $\tilde{\D}$ is indeed a monotone distribution, as claimed.

    \noindent It only remains to bound $\totalvardist{\D}{\tilde{\D}}$:
    \begin{align*}
      2\totalvardist{\D}{\tilde{\D}} &= \sum_{j=1}^{k} \abs{ \D(I_j) - \tilde{\D}(I_j) } = \sum_{j=1}^{k} \abs{ (1-\lambda)D(I_j) - \lambda\eps \sum_{i=j}^{k-1} (2+c) } \\
      &\leq (1-\lambda)\sum_{j=1}^{k} \D(I_j) + \lambda\eps \sum_{j=1}^{k}\sum_{i=j}^{k-1}(2+c) = 1-\frac{1-\eps(2+c)\frac{k(k-1)}{2}}{1+\eps(2+c)\frac{k(k-1)}{2}  }.
    \end{align*}
    Finally, note that $\tilde{\D}$ is a mixture of $\D$ (with weight $\lambda$) and an explicit arithmetically non-increasing distribution; sampling from $\tilde{\D}$ is thus straightforward, and needs at most one sample from $\D$ for each draw.
    \end{proof}
    \begin{remark}
    The above scheme can be easily adapted to the case where the ratio between consecutive intervals is not always the same, but is instead $\abs{I_{j+1}}/\abs{I_j} = 1+c_j$ for some known $c_j\in[C_1,C_2]$; the result 
    then depends on the ratio  $C_2/C_1=\bigTheta{1}$ as well.
    \end{remark}

\else     \begin{proof}[Proof idea]
      We first argue that no interval can have significantly more total weight than the previous one, as it would otherwise contradict the bound on the closeness to monotonicity. This bound on the ``jump'' between two consecutive intervals enables us to define a new distribution $\hat{\D}$ which is a mixture of $\D$ with an arithmetically decreasing $k$-histogram (which only depends on $\eps$ and $k$); it can be shown that for the proper choice of parameters, $\hat{\D}$ is now monotone. Full details can be found in \autoref{appendix:misc:proofs:monotonicity}.
    \end{proof}
\fi 
As a direct corollary, this describes how to correct distributions which are promised to be (very) close to monotone, in a completely \emph{oblivious} fashion:  that is, the behavior of the corrector does not depend on what the input distribution is; furthermore, the probability of failure is null (i.e.,  $\delta=0$). 
\begin{corollary}[Oblivious correcting of monotonicity]\label{lemma:oblivious:corrector:monotonicity:weak}
For every $\eps^\prime\in(0,1)$, there exists an (oblivious) sampling corrector for monotonicity, with parameters $\eps = \bigO{{\eps^\prime}^3/\log^2 n}$, $\eps_1=\eps^\prime$ and sample complexity $\bigO{1}$.
\end{corollary}
\begin{proof}
We will apply \autoref{theorem:oblivious:corrector} for $k=\bigO{\log n/\eps^\prime}$ and $\mathcal{I}$ being the corresponding Birg\'e decomposition (with parameter $\eps^\prime/2$). The idea is then to work with the ``flattening'' $\bar{\D}$ of $\D$: since $\D$ is \eps-close to monotone, it is also $(\eps^\prime/2)$-close, and $\bar{\D}$ is both $(\eps^\prime/2)$-close to $\D$ and \eps-close to monotone. Applying the correcting scheme with our value of $k$ and $c$ set to $\eps^\prime$, the corrected distribution $\tilde{\D}$ is monotone, and 
\[
	\totalvardist{\bar{\D}}{\tilde{\D}} \leq 1-\frac{1-\eps(2+\eps^\prime)\frac{k(k-1)}{2}}{1+\eps(2+\eps^\prime)\frac{k(k-1)}{2} } \leq \frac{\eps^\prime}{2}
\]
where the last inequality derives from the fact that $k^2\eps = \bigO{\eps^\prime}$. This in turn implies by a triangle inequality that $\tilde{\D}$ is $\eps^\prime$-close to $\D$. Finally, observe that, as stated in the lemma, $\tilde{\D}$ can be easily simulated given access to $\D$, using either $0$ or $1$ draw: indeed, $\tilde{\D}$ is a mixture with known weights of an explicit distribution and $\bar{\D}$, and access to the latter can be obtained from $\D$. \end{proof}

\begin{remark}\label{rk:oblivious:stronger:queries}
An interesting feature of the above construction is that does not
\emph{only} yields a $\bigO{1}$-query corrector from \new{sampling to sampling}: it
similarly implies a corrector from \ORACLE to \ORACLE with query complexity
$\bigO{1}$, for \ORACLE being (for instance) an evaluation or \Cdfsamp oracle (cf. \autoref{appendix:definitions}). 
This follows from the fact that the corrected distribution $\tilde{\D}$ is of the form $\tilde{\D} = \lambda\D+(1-\lambda)P$, where
both $\lambda$ and $P$ are fully known.
\end{remark}

\subsection{Correcting with \Cdfsamp{access}}\label{sec:monotonicity:cdf}
\newcommand{\supb}{S}

In this section we prove the following result, which shows that correcting monotonicity with $\littleO{\log n}$ queries (on expectation) is possible when one allows a stronger type of access to the original distribution. In particular, recall that in the \Cdfsamp model (as defined in \autoref{appendix:definitions}) the algorithm is allowed to make, in addition to the usual draws from the distribution, \emph{evaluation queries} to its cumulative distribution function.\footnote{We remark that our algorithm will in fact only use this latter type of access, and will not rely on its ability to draw samples from $\D$.}

\begin{theorem}\label{theo:samp:corrector:monotonicity:cdf}
For any $\eps \in (0,1]$, any number of queries $m$ and $\eps_1 = \bigO{\eps}$ as in the definition, there exists a sampling corrector for monotonicity from \Cdfsamp to \SAMP with \emph{expected} sample complexity $\bigO{\sqrt{m\log n/\eps}}$.
\end{theorem}
\noindent In particular, since learning distributions in the \Cdfsamp model is easily seen to have query complexity $\bigTheta{\log n/\eps}$ (e.g., by considering the lower bound instance of~\cite{Birge:87}), the above corrector beats the ``learning approach'' as long as $m=\littleO{\log n/\eps}$.

\newer{
\begin{remark}
One may look at this ability to correct up to $\littleO{\log n/\eps}$ samples cautiously, as it is well-known that the lower bound for testing monotonicity of distributions is $\bigOmega{\sqrt{n}/\eps^2}$ already~\cite{BKR:04}. However, this lower bound only establishes a \emph{worst-case} indistinguishability: as pointed out in~\autoref{rk:samples:lb:pt}, for many ``typical'' distributions that are \eps-close to monotone, as few as $O(1/\eps)$ samples would be sufficient to detect the discrepancy from monotone (and compromise the correctness of any algorithm relying on these uncorrected samples). 
\end{remark}
}

\subsubsection{Overview and discussion}

A natural idea would be to first group the elements into consecutive intervals (the ``buckets''), and correct this distribution (now a histogram over these buckets) at two levels. That is, start by correcting it optimally at a coarse level (the ``superbuckets,'' each of them being a group of consecutive buckets); then, every time a sample has to be generated, draw a superbucket from this coarse distribution and correct at a finer level \emph{inside this superbucket}, before outputting a sample from the corrected local distribution \new{(i.e. conditional on the superbucket that was drawn and corrected)}.
While this approach seems tantalizing, the main \newer{difficulty} with it lies in the possible boundary violations between superbuckets: that is, even if the average weights of the superbuckets are non-increasing, and the distribution over buckets is non-decreasing inside each superbucket, it might still be the case that there are local violations between adjacent superbuckets. (I.e., the boundaries are bad.) A simple illustration is the sequence $\langle .5,.1,.3,.1 \rangle$, where the first ``superbucket'' is $(.5,.1)$ and the second $(.3,.1)$. The average weight is decreasing, and the sequence is locally decreasing inside each superbucket; yet overall the sequence is not monotone.\smallskip

\noindent Thus, we have to consider 3 kinds of violations:
\begin{enumerate}[(i)]
  \item\label{item:violation:1} global superbucket violations: the average weight of the superbuckets is not monotone.
  \item\label{item:violation:2} local bucket violations: the distribution of the buckets inside some superbucket is not monotone.
  \item\label{item:violation:3} \newer{superbucket boundary violations}: the probability of the last bucket of a superbucket is lower than the probability of the first bucket of the next superbucket.
\end{enumerate}

The ideas underlying our sampling corrector (which is granted both sampling and cumulative query access to the distribution, as defined in the \Cdfsamp access model) are quite simple: after reducing \textit{via} standard techniques the problem to that of correcting a \emph{histogram} supported of logarithmically many intervals (the ``Birg\'e decomposition''), we group these $\ell$ intervals in $K$ ``superbuckets,'' each containing $L$ consecutive intervals from that histogram (``buckets''). (As a guiding remark, our overall goal is to output samples from a corrected distribution using $o(\ell)$ queries, as otherwise we would already use enough queries to actually learn the distribution.) This two-level approach will allow us to keep most of the corrected distribution implicit, only figuring out (and paying queries for that) the portions from which we will ending up outputting samples.

By performing $K$ queries, we can exactly learn the coarse distribution on superbuckets, and correct it for monotonicity \newer{(optimally, e.g. by a linear program ensuring the average weights of the superbuckets are monotone)}, solving the issues of type~\ref{item:violation:1}. In order to fix the boundary violations~\ref{item:violation:3} on-the-go, the idea is to allocate to each superbucket an extra \emph{budget} of probability weight that \emph{can} be used for these boundary corrections. \newer{Importantly}, if this budget is not entirely used the sampling process restarts from the beginning with a probability corresponding with the remaining budget. This effectively ends up simulating a distribution where each superbucket was assigned an extra weight matching exactly what was needed for the correction, \emph{without having to figure out all these quantities beforehand} (as this would cost too many queries). 

Essentially, each superbucket is selected according to its ``potential weight,'' that includes both the actual probability weight it has and the extra budget it is allowed to use for corrections. Whenever a superbucket $S_i$ is selected this way, we first perform optimal local corrections of type \ref{item:violation:2} both on it and the previous superbucket $S_{i-1}$ \new{making a \cdf query at every boundary point between buckets in order to get the} weights of all $2L$ buckets they contain, and \new{then} computing the optimal fix: at this point, the distribution is monotone inside $S_i$ (and inside $S_{i-1}$). After this, we turn to the possible boundary violations of type \ref{item:violation:3} between $S_{i-1}$ and $S_{i}$, by ``pouring'' some of the weight from $S_{i}$'s budget to fill ``valleys'' in the last part of $S_{i-1}$. Once this water-filling has ended,\footnote{\newest{We borrow this graphic analogy with the process of pouring water from~\cite{ACCL:08}, which employs it in a different context (in order to bound the running time of an algorithm by a potential-based argument.).}} we may not have used all of $S_i$'s budget (but as we shall see we make sure we never run out of it): the remaining portion is thus redistributed to the whole distribution by restarting the sampling process from the beginning with the corresponding probability. Note that as soon as we know the weights of all $2L$ Birg\'e buckets, no more \cdf queries are needed to proceed.

\subsubsection{Preliminary steps (preprocessing)}\label{sec:mon:corr:prelim}

\begin{description}
  \item[First step: reducing $\D$ to a histogram.] 
Given \cdfsamp access (i.e., granting both \SAMP and cumulative distribution function (\cdf) query access) to an unknown distribution $\D$ over $[n]$ which is \eps-close to monotone, we can simulate \cdfsamp access to its Birg\'e flattening $\D^{(1)} \eqdef \birge[\D]{\eps}$, also \eps-close to monotone and $3\eps$-close to $\D$ (by~\autoref{coro:Birge:decomposition:robust}). For this reason, we hereafter work with $\D^{(1)}$ instead of $\D$, as it has the advantage of being an $\ell$-histogram for $\ell=\bigO{{\log n}/{\eps}}$. Because of this first reduction, it becomes sufficient to perform \cdf queries on the buckets (and not the individual elements of $[n]$), which
altogether entirely define $\D^{(1)}$.

  \item[Second step: global correcting of the superbuckets.] By making $K$ \cdf queries, we can figure out exactly the quantities $\D^{(1)}(\supb_1), \dots, \D^{(1)}(\supb_K)$. By running a linear program, we can re-weight them to obtain a distribution $\D^{(2)}$ such that \textsf{(a)} the averages $\frac{\D^{(2)}(\supb_j)}{\abs{\supb_j}}$ are non-increasing; \textsf{(b)} the conditional distributions of $\D^{(1)}$ and $\D^{(2)}$ on each superbucket are identical (${\D^{(2)}}_{\supb_j} = {\D^{(1)}}_{\supb_j}$ for all $j\in[K]$); and \textsf{(c)} $\sum_j \abs{ \D^{(2)}(\supb_j) - \D^{(1)}(\supb_j)}$ is minimized.

  \item[Third step: allocating budgets to superbuckets.] For reasons that will become clear in the subsequent, ``water-filling'' step, we want to give each superbucket $\supb_j$ a budget $b_j$ of ``extra weight'' \emph{added to its first bucket $\supb_{j,1}$} that can be used for local corrections when needed -- if it uses only part of this budget during the local correction, it will need to ``give back'' the surplus. To do so, define $\D^{(3)}$ as the distribution such that
\begin{itemize}
  \item $\D^{(3)}(\supb_j) = \lambda^{(3)} (\D^{(2)}(\supb_j) + b_j)$, $j\in [K]$ (where $b_j\eqdef \D^{(2)}(\supb_{j})/(1+\eps)$ for $j \in [K]$; and  $\lambda^{(3)} \eqdef (1 + \sum_j b_j)^{-1}$ is a normalization factor). Note that $\sum_j b_j = 1/(1+\eps) \in [1/2,1]$, so that $\lambda^{(3)}\in[1,2]$.
  \item \new{The conditional distribution on $S_j\setminus S_{j,1}$ satisfy} ${\D^{(3)}}_{\supb_j\setminus \supb_{j,1}} = {\D^{(2)}}_{\supb_j\setminus \supb_{j,1}}$ for all $j\in[K]$.
\end{itemize} 
That is, $\D^{(3)}$ is a version of $\D^{(2)}$ where each superbucket is re-weighted, but ``locally'' looks the same inside each superbucket \emph{except for the first bucket of each superbucket, that received the additional ``budget weight.''} Observe that since the size $\abs{\supb_j}$ of the superbuckets is multiplicatively increasing by an $(1+\eps)$ factor (as a consequence of Birg\'e bucketing), the averages $\D^{(3)}(S_j)/\abs{S_j}$ will remain non-increasing.     That is, the average changes by less for ``big'' values of $j$'s than for small values, as the budget is spread over more elements.

\begin{remark}
$\D^{(3)}$ is uniquely determined by $\eps,n$ and $\D$, and can be explicitly computed using $K$ \cdf queries.
\end{remark}

\end{description}

\subsubsection{Sampling steps (correcting while sampling)}\label{sec:mon:corr:sampling}
Before going further, we describe a procedure that will be necessary for our fourth step, as it will be the core subroutine allowing us to perform local corrections \emph{between superbuckets}.

\paragraph{Water-filling.}
Partition each superbucket $\supb_i$ into range $H_i$, $M_i$ and $L_i$ where (assuming the buckets in  $\supb_i$ are monotone):
\begin{itemize}[-]
  \item $m_i={\D^{(3)}(\supb_i)}/{\abs{\supb_i}}$ is the initial value of the average value of superbucket $\supb_i$ [this does not change throughout the procedure]
  \item $H_i$ are the (leftmost) elements whose value is greater than $m_i$ \hfill [these elements may move to $M_i$ or stay in $H_i$]
  \item $M_i$ are the (middle) elements whose value is equal to $m_i$    \hfill[these elements stay in $M_i$]
  \item $L_i$ are the (rightmost) elements whose value is less than $m_i$  \hfill[these elements may move to $M_i$ or stay in $L_i$]
  \item $\min_i$ is the minimum probability value in superbucket $\supb_i$   \hfill[this updates throughout the procedure]
  \item $\max_i$ is the maximum probability value in superbucket $\supb_i$   \hfill[this updates throughout the procedure]
\end{itemize}
Let $e_i \eqdef \sum_{x \in H_i} (p(x) - m_i)$ to be the \emph{surplus}  (so that if $e_i=0$ then $H_i=\emptyset$ and the superbucket is said to be \emph{dry}) and $d_i \eqdef \sum_{x \in L_i} (m_i- p(x))$ to be the \emph{deficit} (if $d_i = 0$ then $L_i=\emptyset$ and the superbucket is said to be \emph{full}).

\begin{algorithm}[H]
  \caption{Procedure \textsf{water-fill}}
  \begin{algorithmic}[1]
  \State take an infinitesimal amount $\partial p$ from the top of the max, leftmost buckets
             of $H_{i+1}$, in superbucket $\supb_{i+1}$ (this would
             be from the first bucket and any other buckets that have the same probability)
  \State pour $\partial p$ into superbucket $\supb_i$    (this would land in the min, rightmost buckets of 
           $L_{i}$, in superbucket $\supb_{i}$ and spread to the left, to buckets that have the same probability, 
           just like water)
  \end{algorithmic}
\end{algorithm}
\begin{algorithm}[H]
  \caption{Procedure \textsf{front-fill}}
  \begin{algorithmic}[1]
  \While{ the surplus $e_{i+1}$ is greater than the extra budget $b_{i+1}$ allocated in \autoref{sec:mon:corr:prelim} }
  \State take an infinitesimal amount $\partial p$ from the top of the max, leftmost elements
             of $H_{i+1}$, in superbucket $\supb_{i+1}$ (this would
             be from the first bucket and any other buckets that have the same probability)
  \State pour $\partial p$ into the very first bucket of the domain, $S_{1,1}$.
  \EndWhile
  \State \newest{\Return the total amount $f_i$ of weight poured into $S_{1,1}$.}
  \end{algorithmic}
\end{algorithm}
\begin{algorithm}[H]
  \caption{Procedure \textsf{water-boundary-correction}}
  \begin{algorithmic}[1]
    \Require Superbucket index $j=i+1$, \newest{with initial weight $\D^{(3)}(S_{i+1})$.}
   \State move weight from the surplus of $H_{i+1}$ into $L_{i}$ using \textsf{water-fill} until:
    \begin{enumerate}[(a)]
      \item $\max_{i+1} \leq \min_i$; or
      \item $L_{i}=\emptyset$ ($\supb_i$ is full) -- i.e. $\min_i$ = $m_i$; or
      \item $H_{i+1}=\emptyset$ ($\supb_{i+1}$ is dry) -- this can only happen if $e_{i+1} < d_i$ \Comment{\sf This should never happen because of the ``budget allocation'' step.}
    \end{enumerate}
  \State Note that the distribution might not yet be monotone on $\supb_i\cup\supb_{i+1}$, if one of the last two conditions is reached first. If this is the case, then do further correction:
    \begin{enumerate}[(a)]
      \item if $L_{i}=\emptyset$  then do \textsf{front-fill} until $\max_{i+1} \leq \min_i$  (this will happen before $H_{i+1}=\emptyset$)
      \item if $H_{i+1}=\emptyset$ then abort and \Return \fail \Comment{\sf This should never happen because of the ``budget allocation'' step.}
          \end{enumerate}
  \State \newest{\Return the list $B_1,\dots,B_s$ of buckets in $T_i \eqdef L_i\cup S_{i+1}$, along with the weights $w_1,\dots, w_s$ they have from $w$ after the redistribution and the portion $\eps_i$ of the budget that was not used and the portion $f_i$ that was moved by $\textsf{front-fill}$ (so that $\lambda^{(3)}\eps_i + f_i + \sum_{t=1}^s w_t = \D^{(3)}(S_{i+1})$).}
  \end{algorithmic}
  \end{algorithm}

\paragraph{Sampling procedure.}
Recall that we now start and work with $\D^{(3)}$, as obtained in~\autoref{sec:mon:corr:prelim}.

\begin{itemize}
  \item Draw a superbucket $\supb_{i+1}$ according to the distribution $\D^{(3)}(\supb_1)$, $\D^{(3)}(\supb_K)$ on $[K]$.
    \item If $\supb_{i+1} \neq \supb_1$ (we did not land in the first superbucket):
  \begin{itemize}
    \item Obtain (\textit{via} \cdf queries, if they were not previously known) the $2L$ values  $\D^{(3)}(\supb_{i,j})$ $\D^{(3)}(\supb_{i+1,j})$ ($j\in[L]$) of the buckets in superbuckets $\supb_{i},\supb_{i+1}$.
    \item Correct them (separately for each of the two superbuckets) optimally for monotonicity, e.g. via linear programming (if that was not done in a prior stage of sample generation), \newer{ignoring} the extra budget $b_{i}$ and $b_{i+1}$ on the first bucket of each superbucket. Compute $H_i,M_i,L_i$ and $H_{i+1},M_{i+1},L_{i+1}$.
    \item Call \textsf{water-boundary-correction} on $(i+1)$ \new{using the extra budget only if $S_{i+1}$ becomes dry and not counting it while trying to satisfy condition \textsf{(a)}.}\footnote{At this point, the ``new'' distribution $\D^{(4)}$ (which is at least partly implicit, as only known at a very coarse level over superbuckets and locally for some buckets inside $\supb_i\cup\supb_{i+1}$) obtained is monotone over the superbuckets (\textsf{water-boundary-correction} does not violate the invariant that the distribution over superbuckets is monotone), is monotone inside both $\supb_i$ and $\supb_{i+1}$, and furthermore is monotone over $\supb_i\cup\supb_{i+1}$. Even more important, the fact that $\min_i \geq \max_{i+1}$ will ensure applying the same process in the future, e.g. to $\supb_{i+2}$, will remain consistent with regard to monotonicity.}
  \end{itemize}
  \item \newest{If $\supb_{i+1} = \supb_1$ (we landed in the first superbucket), we proceed similarly as per the steps above, except for the \textsf{water-boundary-correction}. That is, we only correct locally $\supb_1$ for monotonicity.}
  \item During the execution of \textsf{water-boundary-correction}, the water-filling procedure may have used some of the initial ``allocated budget'' $b_{i+1}$ to pour into $L_i$. Let $\eps_i \in [0,b_{i+1}]$ be the amount of the budget remaining (not used).
  \new{ 
  \begin{itemize}
  \item with probability $p_i\eqdef\lambda^{(3)}\eps_i/\D^{(3)}(\supb_{i+1})$, restart the sampling process from the beginning (this is the ``budget redistribution step,'' which ensures the correction only uses \emph{what it needs} for each superbucket).
  \item with probability  $q_i\eqdef f_i/\D^{(3)}(\supb_{i+1})$, where $f_i$ is the weight moved by the procedure \textsf{front-fill}, output \newer{from} the very first \newer{bucket} of the domain.      
    \item \newer{with the remaining probability, output a sample from the new (conditional) distribution on the buckets in $T_i\eqdef L_i\cup S_{i+1}$. This is the conditional distribution defined \newest{on $T_i$  by the weights $w_1,\dots, w_s$, as returned by} \textsf{water-boundary-correction}.}
  \end{itemize}
      } 
\end{itemize}
Note that the distribution we output from if we initially select the superbucket $S_{i+1}$, is supported on $L_i\cup S_{i+1}$. Moreover, conditioning on $M_{i+1}\cup L_{i+1}$ we get exactly the conditional distribution $D^{(3)}_{M_{i+1}\cup \newest{L_{i+1}}}$. (This ensures that from each bucket there is a unique superbucket that has to be picked initially for the bucket's weight to be modified.) \newest{Observe that as defined above, buckets from $L_i\subseteq \supb_i$ can be outputted from either because superbucket $\supb_i$ was picked, or because $\supb_{i+1}$ was drawn and some of its weight was reassigned to $L_i$ by \textsf{water-boundary-correction}. The probability of outputting any bucket in $L_i$ is then the sum of the probabilities of the two types of events.}

\subsubsection{Analysis}

The first observation is that the distribution of any sample output by the sampling process described above is not only consistent, but completely determined by $n$, $\eps$ and $\D$:
\begin{claim}\label{lemma:correct:waterfill:defined}
The process described in \autoref{sec:mon:corr:prelim} and \ref{sec:mon:corr:sampling} uniquely defines a distribution $\tilde{\D}$, which is a function of $\D$, $n$ and $\eps\in(0,1)$ only.
\end{claim}

\begin{claim}\label{lemma:correct:waterfill:queries}
The expected number of queries necessary to output $m$ samples from $\tilde{\D}$ is upper bounded by $K + 4mL\eps$.
\end{claim}
\begin{proof}
The number of queries for the preliminary stage is $K$; after this, generating a sample requires $X$ queries, where $X$ is a random variable satisfying
\[
X \leq 2L + RX^\prime
\]
where $X,X^\prime$ are independent and identically distributed, and $R$ is a Bernoulli random variable independent of $X^\prime$ and with parameter $\Delta$ (itself a random variable depending on $X$: $\Delta$ takes value $p_i$ when the first draw selects superbucket $i+1$), corresponding to the probability of restarting the sampling process from the beginning. It follows that
\[
\expect{X} \leq 2L + \expect{R}\expect{X} = 2L + \expect{\Delta}\expect{X}.
\]
Using the fact that $\expect{\Delta}=\sum_{i\in[K]} \D^{(3)}(\supb_{i+1}) p_i = \sum_{i\in[K]} \D^{(3)}(\supb_{i+1})  \frac{\lambda^{(3)}\eps_i}{\D^{(3)}(\supb_{i+1})} \leq \lambda^{(3)}\sum_{i\in[K]} b_i \in[1/3,1/2]$ 
and rearranging, we get $\expect{X}\leq 4L$.
\end{proof}

\begin{lemma}\label{lemma:correct:waterfill:monotone}
If $\D$ is a distribution on $[n]$ satisfying $\totalvardist{\D}{\mathcal{M}} \leq \eps$, then the distribution $\tilde{\D}$ defined above is monotone.
\end{lemma}
\begin{proof}
Observe that as the average weights of the superbuckets in $\D^{(2)}$ are non-increasing, the definition of $\D^{(3)}$ along with the fact that the lengths of the superbuckets are (multiplicatively) increasing implies that the average weights of the superbuckets in $\D^{(3)}$ are also non-increasing. In more detail, fix $1\leq i \leq K-1$; we have
\[
  \frac{\D^{(2)}(S_{i})}{\abs{S_i}} \geq \frac{\D^{(2)}(S_{i+1})}{(1+\eps)\abs{S_i}}
\]
using the fact that $\abs{S_j} = (1+\eps) \abs{S_{j-1}}$. From there, we get that 
\[
(1+\eps)\D^{(2)}(S_i) \geq \D^{(2)}(S_{i+1})
\] or equivalently 
\[
\frac{b_{i}+\D^{(2)}(S_{i})}{\abs{S_{i}}} = \frac{\D^{(2)}(S_{i})+(1+\eps)\D^{(2)}(S_{i})}{(1+\eps)\abs{S_{i}}} \geq 
\frac{\D^{(2)}(S_{i+1})+(1+\eps)\D^{(2)}(S_{i+1})}{(1+\eps)^2\abs{S_{i}}} =\frac{b_{i+1}+\D^{(2)}(S_{i+1})}{\abs{S_{i+1}}}
\]
showing that before renormalization (and therefore after as well) the average weights of the superbuckets in $\D^{(3)}$ are indeed non-increasing. Rephrased, this means that the sequence of $m_i$'s, for $i\in[K]$, is monotone. 
Moreover, notice that by construction the distribution $\tilde{\D}$ is monotone within each superbucket: indeed, it is explicitly made so one superbucket at a time, in the third step of the sampling procedure. After a superbucket has been made monotone this way, it only be changed by water-filling which by design can never introduce new violations: the weight is always moved ``to the left,'' with the values $m_i$'s acting as boundary conditions to stop the waterfilling process and prevent new violations, or moved to the first element of the domain.

It only remains to argue that monotonicity is not violated at the boundary of two consecutive superbuckets. But since the \textsf{water-boundary-correction}, if it does not abort, guarantees that the distribution is monotone between consecutive buckets as well (as $m_{i+1}\leq \max_{i+1}\leq \min_i \leq m_i$), it it sufficient to show that \textsf{water-boundary-correction} never returns \fail. This is ensured by the ``budget allocation'' step, which by providing $H_{i+1}$ with up to an additional $b_{i+1}$ to spread into $L_i$ guarantees it will become dry. Indeed, if this happened then it would mean that correcting this particular violation (before the budget allocation, which only affects the first elements of the superbuckets) in $\D^{(2)}$ required to move more than $b_{i+1}$ weight, contradicting the fact that the average weights of the superbuckets in $\D^{(2)}$ were non-increasing. In more detail, the maximum amount of weight to ``pour'' in order to fill $L_i$ is in the case where $H_{i+1}$ is empty (i.e., the distribution on $S_{i+1}$ is already uniform) but $L_i$ is (almost) all of $S_i$ (i.e., all the weight in $S_i$ is in the first bucket). To correct this with our waterfilling procedure, one would have to pour $\abs{S_i}\cdot \frac{\D^{(2)}(S_{i+1})}{\abs{S_{i+1}}}=\frac{\D^{(2)}(S_{i+1})}{1+\eps}$ weight in $L_i$, which is exactly our choice of value for $b_{i+1}$.
\end{proof}

\begin{lemma}\label{lemma:correct:waterfill:distance}
If $\D$ is a distribution on $[n]$ satisfying $\totalvardist{\D}{\mathcal{M}} \leq \eps$, then $\totalvardist{\D}{\tilde{\D}} = \bigO{\eps}$.
\end{lemma}
\begin{proof}
We will bound separately the distances $\D$ to $\D^{(1)}$, $\D^{(1)}$ to $\D^{(2)}$ and $\D^{(2)}$ to $\tilde{\D}$, and conclude by the triangle inequality.
\begin{itemize}
  \item First of all, the distance $\totalvardist{\D}{\D^{(1)}}$ is at most $3\eps$, by properties of the Birg\'e decomposition (and as $\totalvardist{\D}{\mathcal{M}} \leq \eps$).
  \item We now turn to $\totalvardist{\D^{(1)}}{\D^{(2)}}$, showing that it is at most $\eps$: in order to do so, we introduce $\D^\prime$, the piecewise-constant distribution obtained by ``flattening'' $\D^{(1)}$ on each of the $K$ superbuckets (so that $\D^\prime(S_j) = \D^{(1)}(S_j)$ for all $j$). It is not hard to see, e.g. by the data processing inequality for total variation distance, that $\D^\prime$ is also \eps-close to monotone, and additionally that the closest monotone distribution $M^\prime$ can also be assumed to be constant on each superbucket.

Consider now the transformation that re-weights in $\D^\prime$ each superbucket $S_j$ by a factor $\alpha_j > 0$ to obtain $M^\prime$; it is straightforward to see from~\autoref{sec:mon:corr:prelim} that this transformation maps $\D^{(1)}$ to $\D^{(2)}$. Therefore,
\begin{align*}
2\totalvardist{\D^{(1)}}{\D^{(2)}} &= \sum_{j\in [K]} \sum_{x\in S_j} \abs{ \D^{(1)}(x) - \D^{(2)}(x) } = \sum_{j\in [K]} \sum_{x\in S_j} \abs{ \D^{(1)}(x) - \alpha_j\D^{(1)}(x) } \\
&= \sum_{j\in [K]} \sum_{x\in S_j}  \D^{(1)}(x)\cdot \abs{ 1 - \alpha_j }= \sum_{j\in [K]}  \D^{(1)}(S_j)\cdot \abs{ 1 - \alpha_j } \\
&= \sum_{j\in [K]} \sum_{x\in S_j}  \D^{(1)}(x)\cdot \abs{ 1 - \alpha_j }= \sum_{j\in [K]}  \D^{(1)}(S_j)\cdot \abs{ 1 - \frac{M^\prime(S_j)}{\D^\prime(S_j)} } \\
&= \sum_{j\in [K]} \abs{ \D^\prime(S_j) - M^\prime(S_j)} = 2\totalvardist{\D^\prime}{M^\prime} \leq 2\eps.
\end{align*}
  \item To bound $\totalvardist{\D^{(2)}}{\tilde{\D}}$, first consider the distribution $\D^{\prime\prime}$  obtained by correcting optimally $\D^{(2)}$ for monotonicity \emph{inside each superbucket separately}. That is, $\D^{\prime\prime}$ is the distribution satisfying
  monotonicity on each $S_j$ (separately) and $\D^{\prime\prime}(S_j) = \D^{(2)}(S_j)$ for each $j\in [K]$;
    and minimizing 
    \[
      \sum_{j\in[K]} \sum_{i\in [L]} \abs{ \D^{\prime\prime}(S_{j,i}) - \D^{(2)}(S_{j,i}) } 
    \]
  (or, equivalently, minimizing $\sum_{i\in [L]} \abs{ \D^{\prime\prime}(S_{j,i}) - \D^{(2)}(S_{j,i}) }$ for all $j\in[K]$). The first step is to prove that $\D^{\prime\prime}$ is close to $\D^{(2)}$: recall first that by the triangle inequality, our previous argument implies that $\D^{(2)}$ is $(2\eps)$-close to monotone. Therefore, the (related) optimization problem asking to find a non-negative function $P$ that minimizes the same objective, but under the different constraints ``$P$ is monotone on $[n]$ and $P([n]) = \D^{(2)}([n])$'' has a solution $P$ whose total variation distance from $\D^{(2)}$ is at most $2\eps$.
  
  But $P$ can be used to obtain $P^\prime$, solution to the original problem, by re-weighting each superbucket $S_j$ the following way:
  \[
    P^\prime(x) \eqdef P(x)\cdot \frac{\D^{(2)}(S_j)}{P(S_j)}, \quad x \in S_j.
  \]
  Clearly, $P^\prime$ satisfies the constraints of the first optimization problem; moreover,
  \begin{align*}
    2\totalvardist{P^\prime}{\D^{(2)} } &= \sum_{j\in[K]} \sum_{x\in S_j} \abs{ P^\prime(x) - \D^{(2)}(x)}
    = \sum_{j\in[K]} \sum_{x\in S_j} \abs{ P(x)\frac{\D^{(2)}(S_j)}{P(S_j)} - \D^{(2)}(x)} \\
    &\leq \sum_{j\in[K]} \sum_{x\in S_j} \abs{ P(x) - \D^{(2)}(x)} + \sum_{j\in[K]} \sum_{x\in S_j} P(x)\abs{ \frac{\D^{(2)}(S_j)}{P(S_j)} - 1}  \\ 
    &= 2\totalvardist{P}{\D^{(2)} } + \sum_{j\in[K]} \abs{\D^{(2)}(S_j) - P(S_j)} \leq 4\totalvardist{P}{\D^{(2)} } \\
    &\leq 8\eps,
  \end{align*}
  where we used the fact that $\sum_{j\in[K]} \abs{\D^{(2)}(S_j) - P(S_j)} = \sum_{j\in[K]} \abs{\sum_{x\in S_j} \left( \D^{(2)}(x) - P(x) \right)} \leq \sum_{j\in[K]} \sum_{x\in S_j} \abs{\D^{(2)}(x) - P(x)}$. As $\totalvardist{P^\prime}{\D^{(2)} }$ is an upperbound on the optimal value of the optimization problem, we get $\totalvardist{\D^{\prime\prime}}{\D^{(2)} } \leq 4\eps$.\smallskip
  
  The next and last step is to bound $\totalvardist{\D^{\prime\prime}}{\tilde{\D}}$, and show that it is $\bigO{\eps}$ as well. To see why this will allow us to conclude, note that $\D^{\prime\prime}$ is the intermediate distribution that the sampling process we follow would define, it there was neither extra budget allocated nor \textsf{water-boundary-correction}. Put differently, $\tilde{\D}$ is derived from $\D^{\prime\prime}$ by adding the ``right amount of extra budget $b^\prime_j \in [0,b_j]$'' to $S_j$, then pouring it to $S_{j-1}$ by waterfilling \newer{and front-filling}; and normalizing afterwards by $(1+\sum_{j\in [K]} b^\prime_j)^{-1}$.
  
  Writing $\tilde{\D}^{\prime\prime}$ for the result of the transformation above before the last renormalization step, we can bound $\totalvardist{\D^{\prime\prime}}{\tilde{\D}}$ by
  \begin{align*}
  2\totalvardist{\D^{\prime\prime}}{\tilde{\D}} &= \normone{ \D^{\prime\prime} - \tilde{\D} } \leq \normone{\D^{\prime\prime} - \tilde{\D}^{\prime\prime}} + \normone{\tilde{\D}^{\prime\prime} - \tilde{\D}} \\
  &\leq \sum_{j\in [K]} b^\prime_j + \sum_{j\in [K]} f_j + \sum_{x\in[n]} \abs{ \Big(1+\sum_{j\in [K]} b^\prime_j\Big)\tilde{\D}(x)-\tilde{\D}(x) } \\
  &\leq \sum_{j\in [K]} b^\prime_j + \sum_{j\in [K]} f_j + \dabs{\Big(1+\sum_{j\in [K]} b^\prime_j\Big)-1 } = 2\sum_{j\in [K]} b^\prime_j + \sum_{j\in [K]} f_j
  \end{align*}
  \newer{where $f_j\geq 0$ is defined as the amount of weight moved from $H_j$ to the first element of the domain during the execution of \textsf{water-boundary-correction}, if \textsf{front-fill} is called,} and the bound on $\normone{\D^{\prime\prime} - \tilde{\D}^{\prime\prime}}$ comes from the fact that $\tilde{\D}^{\prime\prime}$ pointwise dominates $\D^{\prime\prime}$, and has a total additional $\sum_{j\in [K]} b^\prime_j$ weight.
  
  It then suffices to bound the quantities  $\sum_{j\in [K]} f_j$ and $\sum_{j\in [K]} b^\prime_j$, using for this the fact that by the triangle inequality $\D^{\prime\prime}$ is itself $(6\eps)$-close to monotone. The at most $K$ intervals where $\D^{\prime\prime}$ violates monotonicity (which are fixed by using the $b^\prime_j$'s) are disjoint, and centered at the boundaries between consecutive superbuckets: i.e., each of them is in a interval $V_j\subseteq L_{j-1}\cup H_j \subsetneq S_{j-1}\cup S_j$. Because of this disjointness, each transformation of $\D^{\prime\prime}$ into a monotone distribution must add weight in $V_j\cap L_{j-1}$ or subtract some from $V_j\cap H_j$ to remove the corresponding violation.\cmargin{is it detailed enough?}\ By definition of $b^\prime_j$ (as minimum amount of additional weight to bring to $L_{j-1}$ when spreading weight from $H_j$ to $L_{j-1}$), this implies that any such transformation has to ``pay'' at least $b^\prime_j/2$ (in total variation distance) to fix violation $V_j$. From the bound on $\totalvardist{\D^{\prime\prime}}{\mathcal{M}}$, we then get $\sum_{j\in[K]} b^\prime_j \leq 12\eps$. A similar argument shows than $\sum_{j\in [K]} f_j \leq 12\eps$ as well, which in turn yields
  $
    \totalvardist{\D^{\prime\prime}}{\tilde{\D}} \leq 18\eps.
  $
  
  \item Putting these bounds together, we obtain
  \begin{align*}
    \totalvardist{\D}{\tilde{\D}} &\leq \totalvardist{\D}{\D^{(1)}} + \totalvardist{\D^{(1)}}{\D^{(2)}}
    + \totalvardist{\D^{(2)}}{\D^{\prime\prime}} + \totalvardist{\D^{\prime\prime}}{\tilde{\D}} \\
    &\leq 3\eps + \eps + 4\eps + 18\eps = 26\eps.
  \end{align*}
\end{itemize}
\end{proof}
\noindent We are finally in position of proving the main result of the section:
\begin{proofof}{\autoref{theo:samp:corrector:monotonicity:cdf}}
The theorem follows from \autoref{lemma:correct:waterfill:defined}, \autoref{lemma:correct:waterfill:queries}, \autoref{lemma:correct:waterfill:monotone} and \autoref{lemma:correct:waterfill:distance}, setting $K=mL=\sqrt{m\ell}$ (where $\ell=\bigO{\log n/\eps}$ as defined in the Birg\'e decomposition).
\end{proofof}

\section{Constrained Error Models}\label{sec:specific:errors}
	In the previous sections, no assumption was made on the form of the error, only on the amount. 
In this section, we suggest a model of errors capturing the deletion of a whole ``chunk'' of the 
distribution. We refer to this model as the {\em missing data model}, where we assume that some $\eps$ probability is removed by taking out all the weight of an arbitrary interval $[i,j]$ for $1\leq i<j\leq n$
and redistributing it on the rest of the domain as per rejection sampling.\footnote{ That is, if $\D$ was the original distribution, 
the faulty one $\D^{(i,j)}$ is formally defined as $(1+\eps)\indicSet{[n]\setminus[i,j]}\cdot\D - \eps\cdot\uniform_{[i,j]}$, where $\eps=\D([i,j])$.}
We show that one can design sampling improvers for monotone distributions
with arbitrarily large amounts of error. Hereafter, $\D$ will denote the original (monotone) distribution (before the deletion error occured), and $\D^\prime=\D^{i,j}$ the resulting (faulty) one, to which the sampling improver has access.
Our sampling improver follows what could be called the ``learning-just-enough'' approach: instead of attempting to approximate the \emph{whole} unaltered original distribution, it only tries to learn the values of $i,j$; and then generates samples ``on-the-fly.'' At a high level, the algorithm works by \textsf{(i)} detecting the location of the missing interval (drawing a large (but still independent of $n$) number of samples), then \textsf{(ii)} estimating the weight of this interval under the original, unaltered distribution; and finally \textsf{(iii)} filling this gap uniformly by moving the right amount of probability weight from the end of the domain. To perform the first stage, we shall follow a paradigm first appeared in~\cite{DDS:12}, and utilize testing as a subroutine to detect ``when enough learning has been done.'' \medskip

\begin{restatable}{theorem}{missingdatacorrector}\label{theo:monotonicity:missingdata:theo}
For the class of distributions following the ``missing data'' error model, there exists a batch sampling improver \textsc{Missing-Data-Improver} \new{for monotonicity} that, on input $\eps,q,\delta$ and $\alpha$, achieves parameters $\eps_1=\bigO{\eps}$ and any $\eps_2 < \eps$; and has sample complexity $\tildeO{\frac{1}{\eps_2^{3}}\log\frac{1}{\delta}}$ independent of $\eps$. 
\end{restatable}
\noindent The detailed proof of our approach, as well as the description of \textsc{Missing-Data-Improver}, are given in the next subsection.
  
\subsection{Proof of \autoref{theo:monotonicity:missingdata:theo}}
Before describing further the way to implement our 3-stage approach, we will need the following lemmata. The first examines the influence of adding or removing probability weight $\eps$ from a distribution, as it is the case in the missing data model:
\begin{lemma}\label{lemma:error:add}
Let $\D$ be a distribution over $[n]$ and $\eps>0$. Suppose $\D^\prime\eqdef(1+\eps) \D - \eps \D_1$, for some distribution $\D_1$. Then $\totalvardist{\D}{\D^\prime} \leq \eps$. 
\end{lemma}
\noindent The proof follows from a simple application of the triangle inequality to the $\lp[1]$ distance between $\D$ and $\D^\prime$. We note that the same bound applies if $\D^\prime=(1-\eps)\D+\eps \D_1$.\medskip

The next two lemmata show that the distance to monotonicity of distributions falling into this error model can be bounded in terms of the probability weight right after the missing interval.
\begin{lemma}\label{lemma:consec:intervals:far}
Let $\D$ be a monotone distribution and $\D^\prime=\D^{(i,j)}$ be the faulty distribution. If $\D([j+1,2j-i+1])>\eps$, then $\D^\prime$ is $\eps/2$-far from monotone.
\end{lemma}
\begin{proof}
    Let $L\eqdef j-i$ be the length of the interval where the deletion occurred. 
    Since the interval $[j+1,2j-i+1]$ has the same length as $[i,j]$ and weight $p > \eps$, the average weight of an element is at least $\frac{\eps}{L}$. 
    Every monotone distribution $M$ should also be monotone on the interval $[i,2j-i+1]$: therefore, one must have $M([i,j])\geq M([j+1,2j-i+1])$. Let $q\eqdef M([i,j])$. As $\D^\prime([i,j])=0$, we get that $2\totalvardist{\D^\prime}{\tilde{\D}}\geq q$. On one hand, if $q < p$ then at least $q-p$ weight must have been ``removed'' from $[i,2j-i+1]$ to achieve monotonicity, and altogether $2\totalvardist{\D^\prime}{M}\geq q+(p-q)=p$. On the other hand, if $q \geq p$ we directly get $2\totalvardist{\D^\prime}{M}\geq q\geq p$. In both cases, \[
    \totalvardist{\D^\prime}{M}\geq p/2\geq \eps/2
    \]
    and $\D^\prime$ is $\eps/2$-far from monotone.
\end{proof}

\begin{lemma}\label{lemma:consec:intervals:close}
  Let $\D$ be a monotone distribution and $\D^\prime=\D^{(i,j)}$ as above. If $\D^\prime([j+1,2j-i+1])<\eps/2$, then $\D^\prime$ is $\eps$-close to monotone.\end{lemma}
\begin{proof}
  We will constructively define a monotone distribution $M$ which will be $\eps$-close to $\D^\prime$. Let $p\eqdef\D^\prime([j+1,2j-i+1])<\eps/2$. According to the missing data model, $\D^\prime$ should be monotone on the intervals $[1,i-1]$ and $[j+1,n]$. In particular, the probability weight of the last element of $[j+1,2j-i+1]$ should be below the average weight of the interval, i.e. for all $k \geq 2j-i+1$ one has $\D^\prime(k) \leq \D^\prime(2j-i+1)<\frac{p}{j-i+1}$.

  So, if we let the distribution $M$ (that we are constructing) be uniform on the interval $[j+1,2j-i+1]$ and have also total weight $p$ there, monotonicity will not be violated at the right endpoint of the interval; and the $\lp[1]$ distance between $\D^\prime$ and $M$ in that interval will be at most $2p$. ''Taking'' another $p$ probability weight from the very end of the domain and moving it to the interval $[i,j]$ (where it is then uniformly spread) to finish the construction of $M$ adds at most another $2p$ to the $\lp[1]$ distance. Therefore,
  $2\totalvardist{\D^\prime}{M}\leq 2p +2p < 2\eps$; and $M$ is monotone as claimed.  
\end{proof}

\noindent The sampling improver is described in \autoref{algo:missing:data:improver}.
\begin{algorithm}\caption{\textsc{Missing-Data-Improver}}\label{algo:missing:data:improver}
\algblock[block]{Start}{End}
\begin{algorithmic}[1]
  \Require $\eps$, $\eps_2 < \eps$, $\delta \in (0,1)$ and $q\geq 1$, sample access to $\D^\prime$.
  \Start\Comment{\textsc{Preprocessing}}
    \State Draw $m\eqdef\tildeTheta{\frac{1}{\eps_2^3}\log\frac{1}{\delta}}$ samples from $\D^\prime=\D^{i,j}$.
    \State Run the algorithm of \autoref{lemma:monotonicity:missingdata:step1} on them to get an estimate $(a,b)$ of the unknown $(i,j)$ or the value \textsf{close}.
    \State Run the algorithm of \autoref{lemma:monotonicity:missingdata:step2} on them to get an estimate $\gamma$ of $\D^\prime([b+1,2b-a+1])$, and     values $c,\gamma^\prime$ such that $\abs{\D^\prime([c,n])-\gamma^\prime}\leq\eps_2^{3/2}$.
  \End
  \Start\Comment{\textsc{Generating}}
    \For{$i$ from $1$ to $q$}
      \State Draw $s_i$ from $\D^\prime$.
      \If{the second step of \textsc{Preprocessing} returned \textsf{close}, or $\gamma < 5\eps_2^{3/2}$}\label{algo:step:already:close}
        \State \Return $s_i$ \Comment{The distribution is already $\eps_2$-close to monotone; do not change it.}
      \EndIf
      \If{$s_i\in [c,n]$}\Comment{Move $\gamma$ weight from the end to $[a,b]$}
        \State With probability $\gamma/\gamma^\prime$, \Return a uniform sample from $[a,b]$
        \State Otherwise, \Return $s_i$
      \ElsIf{$s_i\in [b+1,2b-a+1]$}
        \State \Return a uniform sample from $[b+1,2b-a+1]$ 
      \Else
        \State \Return $s_i$ \Comment{Do not change the part of $\D^\prime$ that need not be changed.}
      \EndIf
    \EndFor
  \End
\end{algorithmic}
\end{algorithm}

\paragraph{Implementing \textsf{(i)}: detecting the gap}

\begin{lemma}[Lemma \textsf{(i)}]\label{lemma:monotonicity:missingdata:step1}
There exists an algorithm that, on input $\alpha \in(0,1/3)$ and $\delta\in(0,1)$, takes $\tildeO{\frac{1}{\alpha^6}\log\frac{1}{\delta}}$ samples from $\D^\prime=\D^{i,j}$ and outputs either two elements $a,b\in[n]$ or \textsf{close} such that the following holds. With probability at least $1-\delta$, 
\begin{itemize}
  \item if it outputs elements $a,b$, then \textsf{(a)} $[i,j]\subseteq[a,b]$ and \textsf{(b)} $\D^\prime([a,b]) \leq 3\alpha^2$;
  \item if it outputs \textsf{close}, then $\D^\prime$ is $\alpha^2$-close to monotone.
\end{itemize}
\end{lemma}
\begin{proof}
    Inspired by techniques from~\cite{DDS:12},
  we first partition the domain into $t=\bigO{1/\alpha^2}$ intervals 
  $I_1,\dots, I_t$ of roughly equal weight as follows. By taking $\bigO{\frac{1}{\alpha^6}\log\frac{1}{\delta}}$ samples, the DKW inequality ensures that with probability at least $1-\delta/2$ we obtain an approximation $\hat{\D}$ of $\D^\prime$, close up to $\alpha^3/5$ in Kolmogorov distance. We hereafter assume this holds. For our partitioning to succeed, we first have to take care of the ``big elements,'' which by assumption on $\D^\prime$ (which originates from a monotone distribution) must all be at the beginning. In more detail, let 
  \[
    r \eqdef \max\setOfSuchThat{ x\in [n] }{ \hat{\D}(x) \geq \frac{4\alpha^3}{5} }
  \]
  and $B\eqdef\{1,\dots,r\}$ be the set of potentially big elements. Note that if $\D^\prime(x) \geq \alpha^3$, then necessarily $x\in B$. This leaves us with two cases, depending on whether the ``missing data interval'' is amidst the big elements, or in the tail part of the support.  
  \begin{itemize}
    \item If $[i,j] \subseteq B$: it is then straightforward to \emph{exactly} find $i,j$, and output them as $a,b$. Indeed all elements $x\in B$ have, by monotonicity, either $\D^\prime(x) \geq \D^\prime(r) \geq \frac{3\alpha^3}{5}$, or $\D^\prime(x)=0$ (the latter if and only if $x\in[i,j]$). Thus, one can distinguish between $x\in[i,j]$ (for which $\hat{\D}(x) \leq \alpha^3/5$) and $x\notin[i,j]$ (in which case $\hat{\D}(x) \geq 2\alpha^3/5$).
    \item If $[i,j] \not\subseteq B$: then, as $r\notin [i,j]$ (since $\D^\prime(r) > 0$), it must be the case that $[i,j] \subseteq \bar{B}=\{r+1,\dots, n\}$. Moreover, every point $x\in \bar{B}$ is ``light:'' $\D^\prime(x) < \alpha^3$ and $\hat{\D}(x) < \frac{4\alpha^3}{5}$.     We iteratively define $I_1,\dots, I_t\subseteq \bar{B}$, where $I_i=[r_i+1,r_{i+1}]$: $r_1\eqdef r+1$, $r_{t+1}\eqdef n$, and for $1 \leq i \leq t-1$
    \[
        r_{i+1} \eqdef \min\setOfSuchThat{ s > r_i }{ \hat{\D}([r_i+1,x]) \geq \alpha^2 }\;.
    \]
    This guarantees that, for all $i\in[t]$, $\D^\prime( I_i ) \in [\alpha^2-\frac{2\alpha^3}{5}, \alpha^2+\frac{4\alpha^3}{5}+\frac{2\alpha^3}{5}]\subset[\alpha^2-\frac{3\alpha^3}{2},\alpha^2+\frac{3\alpha^3}{2}]$. (And in turn that $t=\bigO{1/\alpha^2}$ as claimed.)
  Observing that the definition of the missing data error model implies 
  $\D^\prime$ is 2-modal, we can now use the monotonicity tester of \cite[Section 3.4]{DDS:12}. This algorithm takes 
  only \new{$\bigO{\frac{k}{\eps^2}\log\frac{1}{\delta}}$} samples 
  (crucially, no dependence on $n$) to distinguish with probability at least $1-\delta$ whether
  a $k$-modal distribution is monotone versus $\eps$-far from it.

  We iteratively apply this tester with parameters $k=2$, $\eps=\alpha^2/4$ and $\delta^\prime=\bigO{\delta/t}$, to each of the at most $t$ prefixes of the form $P_\ell\eqdef \cup_{i=1}^\ell I_i$; a union bound ensures that with probability at least $1-\delta/2$ all tests are correct. Conditioning on this, we are able to detect the first interval $I_{\ell^\ast}$ which either contains or falls after $j$ (if no such interval is found, then the input distribution is already $\alpha^2$-close to monotone and we output \textsf{close}).   
  In more detail, suppose first no run of the tester rejects (so that \textsf{close} is outputted). Then, by \autoref{lemma:consec:intervals:far}, we must have $\D([j+1,2j-i+1])\leq 2\cdot\alpha^2/4=\alpha^2/2$, and \autoref{lemma:consec:intervals:close} guarantees $\D^\prime$ is then $\alpha^2$-close to monotone.
  
  Suppose now that it rejects on some prefix $P_{\ell^\ast}$ (and accepted for all $\ell < \ell^\ast$). As $\D^\prime$ is non-increasing on $[1,j]$, we must have $[i,j]\subset P_{\ell^\ast}$. Moreover, the tester will by \autoref{lemma:consec:intervals:far} reject as soon as an interval $[j+1,s]\subseteq [j+1,2j-i+1]$ of weight $\alpha^2/2$ is added to the current prefix. This implies, as each $I_\ell$ has weight at least $\alpha^2/2$, that $[i,j]\subseteq I_{\ell^\ast-1}\cup {\ell^\ast}=[a,b]$.
  
  Finally, observe that the above can be performed with $\bigO{\frac{1}{\alpha^2}\cdot \frac{1}{\alpha^4} \cdot \log t}=\tildeO{\frac{1}{\alpha^6}\log\frac{1}{\delta}}$ samples, as claimed (where the first $1/\alpha^2$ factor comes from doing rejection sampling to run the tester with domain $P_\ell$ only, which by construction is guaranteed to have weight $\bigOmega{1/\alpha^2}$). The overall probability of failure is at most $\delta/2+\delta/2=\delta$, as claimed.
  \end{itemize} 
\end{proof}

\paragraph{Implementing \textsf{(ii)}: estimating the missing weight}

Conditioning on the output $a,b$ of \autoref{lemma:monotonicity:missingdata:step1} being correct, the next lemma explains how to get a good estimate of the total weight we should \emph{put back} in $[a,b]$ in order to fix the deletion error.
\begin{lemma}\label{lemma:monotonicity:missingdata:step2}
Given $\D^\prime$, $\alpha$ as above, $\delta\in(0,1)$ and $a,b$ such that $[i,j]\subseteq[a,b]$ and $\D^\prime([a,b]) \leq 3\alpha^2$, there exists an algorithm which takes $\bigO{\frac{1}{\alpha^6}\log\frac{1}{\delta}}$ samples from $\D^\prime$ and outputs values $\gamma, \gamma^\prime$ and $c$ such that the following holds with probability at least $1-\delta$: 
\begin{enumerate}[(i)]
  \item\label{lemma:monotonicity:missingdata:step2:1}
         $\abs{ \D^\prime([b+1,2b-a+1])-\gamma } \leq \alpha^3$;
  \item\label{lemma:monotonicity:missingdata:step2:2}
         $\abs{ \D^\prime([c,n])-\gamma^\prime } \leq \alpha^3$ and $\gamma^\prime\geq \gamma$;
  \item\label{lemma:monotonicity:missingdata:step2:3}
         $\D^\prime([c,n]) \geq \D^\prime([b+1,2b-a+1]) - 2\alpha^3$ and $\D^\prime([c+1,n]) < \D^\prime([b+1,2b-a+1]) + 2\alpha^3$;
  \item\label{lemma:monotonicity:missingdata:step2:4}
         $\gamma \leq 2\eps+4\alpha^3$.
\end{enumerate}
\end{lemma} 
\begin{proof}
 Again by invoking the DKW inequality, we can obtain (with probability at least $1-\delta$) an approximation $\hat{\D}$ of $\D^\prime$, close up to $\alpha^3/2$ in Kolmogorov distance. This provides us with an estimate $\gamma$ of $\D^\prime([b+1,2b-a+1])$ satisfying the first item (as, for any interval $[r,s]$, $\hat{\D}([r,s])$ is within an additive $\alpha^3/2$ of $\D^\prime([r,s])$). Then, setting 
 \[
    c \eqdef \max \setOfSuchThat{ x\in[n] }{ \hat{\D}([x,n]) \geq \gamma }
 \]
 and $\gamma^\prime \eqdef \hat{\D}([c,n])$, items \ref{lemma:monotonicity:missingdata:step2:2} and \ref{lemma:monotonicity:missingdata:step2:3} follow. The last bound of \ref{lemma:monotonicity:missingdata:step2:4} derives from an argument identical as of \autoref{lemma:consec:intervals:far} and the promise that $\D^\prime$ is $\eps$-close to monotone: indeed, one must then have $\D^\prime([b+1,2b-a+1]) \leq \D^\prime([a,b]) + 2\eps \leq 2\eps + 3\alpha^2$, which with \ref{lemma:monotonicity:missingdata:step2:1} concludes the argument.
\end{proof} 

\noindent To finish the proof of \autoref{theo:monotonicity:missingdata:theo}, we apply the above lemmata with  $\alpha\eqdef\bigTheta{\sqrt{\eps_2}}$; and need to show that the algorithm generates samples from a distribution that is $\eps_2=\bigO{\alpha^2}$-close to monotone. This is done by bounding the error encountered (due to approximation errors) in the following parts of the algorithm:  when estimating the weight $\gamma$ of an interval of equal length adjacent to the interval $[a,b]$, uniformizing its weight on that interval, and estimating the last $\gamma$-quantile of the distribution, in order to move the weight needed to fill the gap from there. 
If we could have perfect estimates of the gap ($[a,b]=[i,j]$), the missing weight $\gamma$ and the point $c$ such that $\D^\prime([c,n])=\gamma$, the corrected distribution would be monotone, as the probability mass function in both the gap and the next interval would be at the same ``level'' (that is, $\frac{\gamma}{b-a+1}$).

By choice of $m$, with probability at least $1-\delta$ the two subroutines of the \textsc{Preprocessing} stage (from \autoref{lemma:monotonicity:missingdata:step1} and \autoref{lemma:monotonicity:missingdata:step2}) behave as expected. We hereafter condition on this being the case. For convenience, we write $I=[a,b]$, $J=[b+1,2b-a+1]$ and $K=[c,n]$, where $a,b,c$ and $\gamma, \gamma^\prime$ are the outcome of the preprocessing phase.

\paragraph{If the test in Line~\ref{algo:step:already:close} passes.}
If the preprocessing stage returned either \textsf{close}, or a value $\gamma < 5\eps_2^{3/2}=5\alpha^3$, then we claim that $\D^\prime$ is already $\bigO{\alpha^2}$-close to monotone. The first case is by correctness of \autoref{lemma:monotonicity:missingdata:step1}; as for the second, observe that it implies $\D^\prime(J) < 6\alpha^3$. Thus, ``putting back'' (from the tail of the support) weight at most $6\alpha^3$ in $[i,j]$ would be sufficient to correct the violation of monotonicity; which yields an $\bigO{\alpha^3}$ upperbound on the distance of $\D^\prime$ to monotone.

\paragraph{Otherwise.} This implies in particular that $\gamma \geq 5\alpha^3$, and thus $\D^\prime(J) \geq 4\alpha^3$. By \autoref{lemma:monotonicity:missingdata:step2} \ref{lemma:monotonicity:missingdata:step2:3}, it is then also the case that $\D^\prime(K) \geq 2\alpha^3$. Then, denoting by $\tilde{\D}$ the corrected distribution, we have
\[
\tilde{\D}(x) = \begin{cases}
 \D^\prime(x) + \frac{\gamma}{\gamma^\prime}\cdot\frac{\D^\prime(K)}{\abs{I}} & \text{ if } x\in I \\
 \frac{\D^\prime(J)}{\abs{J}} & \text{ if } x\in J \\
 \D^\prime(x)\cdot(1-\frac{\gamma}{\gamma^\prime}) & \text{ if } x\in K \\
 \D^\prime(x) & \text{ otherwise.}
\end{cases}
\]
\paragraph{Distance to $\D^\prime$.} From the expression above, we get that
\[
2\totalvardist{\tilde{\D}}{\D^\prime} \leq \frac{\gamma}{\gamma^\prime} \D^\prime(K) +  2\D^\prime(J) + \frac{\gamma}{\gamma^\prime} \D^\prime(K) = 2\left(\frac{\gamma}{\gamma^\prime} \D^\prime(K) + \D^\prime(J) \right).
\] 
From \autoref{lemma:monotonicity:missingdata:step2}, we also know that $\D^\prime(J) \leq \gamma+\alpha^3$, $\D^\prime(K) \leq \gamma^\prime+\alpha^3$ and $\gamma/\gamma^\prime \leq 1$, so that
\[
\totalvardist{\tilde{\D}}{\D^\prime} \leq \frac{\gamma}{\gamma^\prime} (\gamma^\prime+\alpha^3) + \gamma+\alpha^3 \leq 2(\gamma+\alpha^3) \leq 4\eps+10\alpha^3 = \bigO{\eps}.
\] 
(Where, for the last inequality, we used \autoref{lemma:monotonicity:missingdata:step2} \ref{lemma:monotonicity:missingdata:step2:4}; and finally the fact that $\eps_2\leq \eps$).

\paragraph{Distance to monotone.}

Consider the distributions $M$ defined as
\[
M(x) = \begin{cases}
 \D^\prime(x) + \frac{\D^\prime(J)}{\abs{I}} & \text{ if } x\in I \\
 \frac{\D^\prime(J)}{\abs{J}} & \text{ if } x\in J \\
 \D^\prime(x)\cdot\left(1-\frac{\D^\prime(J)}{\D^\prime(K)}\right) & \text{ if } x\in K \\
 \D^\prime(x) & \text{ otherwise.}
\end{cases}
\]
We first claim that $M$ is $\bigO{\alpha^2}$-close to monotone. Indeed, $M$ is monotone on $[a,n]$ by construction (and as $\D^\prime$ was monotone on $[b,n]$). The only possible violations of monotonicity are on $[1,b]$, due to the approximation of $(i,j)$ by $(a,b)$ -- that is, it is possible for the interval $[a,i]$ to now have too much weight, with $M(a-1) < M(a)$. But as we have $\D^\prime([a,b])\leq 3\alpha^2$, the total extra weight of this ``violating bump'' is $\bigO{\alpha^2}$.\medskip

Moreover, the distance between $M$ and $\tilde{\D}$ can be upperbounded by their difference on $J$ and $K$:
\begin{align*}
  2\totalvardist{\tilde{\D}}{M} \leq 2\abs{ \D^\prime(J) - \frac{\gamma}{\gamma^\prime} \D^\prime(K) } \leq 2\alpha^3\frac{1+\frac{\alpha^3}{\D^\prime(K)}}{1-\frac{\alpha^3}{\D^\prime(K)}} \leq 6\alpha^3
\end{align*}
where we used the fact that $\frac{\gamma}{\gamma^\prime} \in \left[\frac{\D^\prime(J)-\alpha^3}{\D^\prime(K) + \alpha^3}, \frac{\D^\prime(J)+\alpha^3}{\D^\prime(K) - \alpha^3}\right]$, and that $\D^\prime(K) \geq 2\alpha^3$. By the triangle inequality, $\tilde{\D}$ is then itself $\bigO{\alpha^2}$-close to monotone. This concludes the proof of \autoref{theo:monotonicity:missingdata:theo}.\qed
 
\section{Focusing on randomness scarcity}\label{sec:focus:randomness}
	\subsection{Correcting uniformity}\label{sec:uniformity}
		In order to illustrate the challenges and main aspects of this section, we shall focus on what is arguably the most natural property of interest, ``being uniform'' 
(i.e. $\property=\{\uniform_n\})$. 
As a first observation, we note that when one is interested in correcting uniformity on an arbitrary domain $\domain$, allowing arbitrary
amounts of
additional randomness makes the task almost trivial: 
by using roughly $\log\abs{\domain}$ random bits per query, 
it is possible to interpolate arbitrarily 
between $\D$ and the uniform distribution.
One can naturally ask whether the same can be achieved 
\emph{while using no -- or very little -- additional randomness besides the draws from the sampling oracle itself.}
As we show below, this is possible, at the price of a slightly worse query complexity. 
We hereafter focus once again on the case $\domain=[n]$, 
and give constructions which achieve different trade-offs between the level of correction (of $\D$ to uniform), 
the fidelity to the original data 
(closeness to $\D$) and the sample complexity. We then show how
to combine these constructions to achieve reasonable performance
in terms of all the above parameters.
In \autoref{sec:unif:subgroup}, we turn to the related problem of correcting uniformity on an (unknown) subgroup of the domain, and extend our results to this setting.
Finally, we discuss the differences and relations with extractors in \autoref{sec:unif:extractors}.

\paragraph*{High-level ideas} The first algorithm we describe (\autoref{lemma:sampling:corrector:uniformity:vneumann}) is a sampling corrector based on a ``von Neumann-type'' approach: by seeing very crudely the distribution $\D$ as a distribution over two points (the first and second half of the support $[n]$), one can leverage the closeness of $\D$ to uniform to obtain with overwhelming probability a sequence of uniform random bits; and use them to generate a uniform element of $[n]$. The drawback of this approach lies in the number of samples required from $\D$: namely, $\tildeTheta{\log n}$. 

The second approach we consider relies on viewing $[n]$ as the Abelian group $\Z_n$, 
and leverages crucial properties of the convolution of distributions. 
Using a robust version of the fact that the uniform distribution is the absorbing element for this operation, we are able to argue that taking a \emph{constant} number of samples from $\D$ and outputting their sum obeys a distribution $\tilde{\D}$ exponentially closer to uniform (\autoref{lemma:sampling:corrector:uniformity}). This result, however efficient in terms of getting closer to uniform, does not guarantee anything non-trivial about the distance $\tilde{\D}$ to the input distribution $\D$. More precisely, starting from $\D$ which is at a distance \eps from uniform, it is possible to end up with $\tilde{\D}$ at a distance $\eps^\prime$ from uniform, but $\eps+\bigOmega{\eps^\prime}$ from $\D$ (see \autoref{claim:counterexample:convolution} for more details). In other terms, this improver does get us closer to uniform, but somehow can \emph{overshoot} in the process, getting too far from the input distribution.

The third improver we describe (in \autoref{lemma:sampling:corrector:uniformity:hybrid}) yields slightly different parameters: it essentially enables one to get ``midway'' between $\D$ and the uniform distribution, and to sample from a distribution $\tilde{\D}$ (almost) $(\eps/2)$-close to \emph{both} the input and the uniform distributions. It achieves so by combining both previous ideas: using $\D$ to generate a (roughly) unbiased coin toss, and deciding based on the outcome whether to output a sample from $\D$ or from the improver of \autoref{lemma:sampling:corrector:uniformity}.

Finally, by ``bootstrapping'' the hybrid approach described above, one can provide sampling access to an improved $\hat{\D}$ both arbitrarily close to uniform \emph{and} (almost) optimally close to the original distribution $\D$ (up to an additive $\bigO{\eps^3}$), as described in \autoref{lemma:sampling:corrector:uniformity:boostrap}. Note that this is at a price of an extra $\log(1/\eps_2)$ factor in the sample complexity, compared to \autoref{lemma:sampling:corrector:uniformity}: in a sense, the price of ``staying faithful to the input data.''

\begin{restatable}[von Neumann Sampling Corrector]{theorem}{unifvncorr}\label{lemma:sampling:corrector:uniformity:vneumann}
  For any $\eps< 0.49$ (and $\eps_1=\eps$) as in the definition, there exists a sampling corrector for uniformity with query complexity $\bigO{\log n(\log\log n + \log(1/\delta))}$ (where $\delta$ is the probability of failure per sample).
\end{restatable}
\ifnum\fulldetails=1 \begin{proof}
Let $\D$ be a distribution over $[n]$ such that $\totalvardist{\D}{\uniform}\leq \eps < 1/2 - c$ for some absolute constant \newer{$c < 1/2$ (e.g., $c=0.49$)}, and let $S_0, S_1$ denote respectively the sets $\{1,\dots,n/2\}$ and $\{n/2+1,\dots,n\}$. The high-level idea is to see a draw from $\D$ as a (biased) coin toss, depending on whether the sample lands in $S_0$ or $S_1$; by applying von Neumann's method, we then can retrieve a truly uniform bit at a time (with high probability). Repeating this $\log n$ times will yield a uniform draw from $[n]$. More precisely, it is immediate by definition of the total variation distance that $\abs{ \D(S_0)-\D(S_1) } \leq 2\eps$, so in particular (setting $p\eqdef \D(S_0)$) we have access to a Bernoulli random variable with parameter $p\in\left[\frac{1}{2}-\eps,\frac{1}{2}+\eps\right]$.

To generate \emph{one} uniform random bit (with probability of failure at most $\delta^\prime=\delta/\log n$), it is sufficient to take in the worst case $m\eqdef \clg{(\log\frac{1}{1-c})^{-1}\log\frac{2}{\delta^\prime}}$ samples, and stop as soon as a sequence $S_0S_1$ or $S_1S_0$ is seen (giving respectively a bit $0$ or $1$). If it does not happen, then the corrector $\textsc{VN--Improver}_n$ outputs \fail; the probability of failure is therefore
\[
   \probaOf{\textsc{VN--Improver}_n\text{ outputs }\fail} = p^{m} + (1-p)^m \leq 2\cdot (1-c)^m \leq \delta^\prime = \frac{\delta}{\log n}.
\]
By a union bound over the $\log n$ bits to extract, $\textsc{VN--Improver}_n$ indeed outputs a uniform random number $s\in[n]$ with probability at least $1-\delta$, using at most $m\log n = \bigO{\log n \log\frac{\log n}{\delta}}$ samples---and, in expectation, only $\bigO{(\log n)/{p}} = \bigO{\log n}$.
\end{proof}
\fi 
\ifnum\fulldetails=1 \new{As previously mentioned, we hereafter work modulo $n$, equating $[n]$ to the Abelian group $(\Z_n,+)$. This convenient (and equivalent) view will allow us to use properties of convolutions of distributions over Abelian groups,\footnote{For more detail on this topic, the reader is referred to \autoref{appendix:convolution:abelian}.} in particular the fact that the uniform distribution on $\Z_n$ is (roughly speaking) an attractive fixed point for this operation. In particular, taking $\D$ to be the (unknown) distribution promised to be \eps-close to uniform, \autoref{fact:convol:uniform:tv} guarantees that by drawing two independent samples $x,y\sim \D$ and computing $z=x+y \mod n$, the distribution of $z$ is $(2\eps^2)$-close to the uniform distribution on $\{0,\dots,n-1\}$. This key observation is the basis for our next result:}
\fi 
\begin{restatable}[Convolution Improver]{theorem}{unifconvolvi}\label{lemma:sampling:corrector:uniformity}
  For any $\eps< \frac{1}{\sqrt{2}}$, $\eps_2$ and $\eps_1=\eps+\eps_2$ as in the definition, there exists a sampling improver for uniformity with query complexity $\bigO{\frac{\log\frac{1}{\eps_2}}{\log\frac{1}{\eps}}}$.
\end{restatable}
\ifnum\fulldetails=1 \begin{proof}
 Extending by induction the observation above to a sum of finitely many independent samples, we get that by drawing $k\eqdef\frac{\log\frac{1}{\eps_2}-1}{\log\frac{1}{\eps}-1}$ independent elements $s_1,\dots,s_k$ from $\D$ and computing
  \[
  s=\left(\sum_{\ell=1}^k s_\ell \mod n\right) + 1 \in [n]
  \]
  the distribution $\tilde{\D}$ of $s$ is $(\frac{1}{2}\left( 2\eps\right)^k)$-close to uniform; and by choice of $k$, $(\frac{1}{2}\left( 2\eps\right)^k)=\eps_2$.  As $\totalvardist{ \D }{ \tilde{\D} } \leq \totalvardist{ \D }{ \uniform } + \totalvardist{ \uniform }{ \tilde{\D} } \leq \eps+\eps_2$, the vacuous bound on the distance between $\D$ and $\tilde{\D}$ is as stated.
\end{proof}

This triggers a natural question: namely, can this ``vacuous bound'' be improved? That is, setting $\eps\eqdef\totalvardist{ \D }{ \uniform }$ and $\D^{(k)}\eqdef \D\convolution\cdots\convolution \D$ ($k$-fold convolution), what can be said about $\totalvardist{ \D }{ \D^{(k)} }$ as a function of $\eps$ and $k$? Trivially, the triangle inequality asserts that
  \[
          \eps - 2^{k-1}\eps^k \leq \totalvardist{ \D }{ \D^{(k)} } \leq \eps + 2^{k-1}\eps^k\;;
  \]
  but can  the right-hand side be tightened further? For instance, one might hope to achieve $\eps$. Unfortunately, this is not the case: even for $k=2$, one cannot get better than $\eps + \bigOmega{\eps^2}$ as an upper bound. Indeed, one can show that for $\eps\in(0,\frac12)$, there exists a distribution $\D$ on $\Z_n$ such that $\totalvardist{ \D }{ \uniform } = \eps$, yet $\totalvardist{ \D }{ \D\convolution \D } = \eps + \frac{3}{4}\eps^2 + \bigO{\eps^3}$ (see \autoref{claim:counterexample:convolution} in the appendix).

\fi 
\begin{restatable}[Hybrid Improver]{theorem}{unifconvolvii}\label{lemma:sampling:corrector:uniformity:hybrid}
  For any $\eps\leq \frac{1}{2}$, $\eps_1=\frac{\eps}{2}+2\eps^3+\eps^\prime$ and $\eps_2=\frac{\eps}{2}+\eps^\prime$, there exists a sampling improver for uniformity with query complexity $\bigO{\frac{\log\frac{1}{\eps^\prime}}{\log\frac{1}{\eps}}}$.
\end{restatable}
\ifnum\fulldetails=1 \begin{proof}
Let $\D$ be a distribution over $[n]$ such that $\totalvardist{\D}{\uniform} = \eps$, and write $d_0$ (resp. $d_1$) for $\D(\{1,\dots,n/2\})$ (resp. $\D(\{n/2+1,\dots,n\})$). By definition, $\abs{d_0-d_1}\leq 2\eps$.
Define the Bernoulli random variable $X$ by taking two independent samples $s_1$, $s_2$ from $\D$, and setting $X$ to $0$ if both land in the same half of the support (both in $\{1,\dots,n/2\}$, or both in $\{n/2+1,\dots,n\}$). It follows that 
$p_0\eqdef \probaOf{X=0} = d_0^2 + d_1^2$ and $p_1\eqdef \probaOf{X=1} = 2d_0d_1$, i.e. $0\leq p_0-p_1 = (d_1-d_2)^2 \leq 4\eps^2$. In other terms, $X\sim \bernoulli{p_0}$ with $\frac{1}{2}\leq p_0\leq \frac{1}{2}+2\eps^2$.\medskip

\noindent Consider now the distribution
\[
    \tilde{\D} \eqdef (1-p_0)D + p_0 D^{(k)} \vspace{-\baselineskip}
\]
where $\D^{(k)}=\overbrace{\D\convolution\dots\convolution \D}^{k\text{ times}}$ as in \autoref{lemma:sampling:corrector:uniformity}. Observe than getting a sample from $\tilde{\D}$ only requires at most $k+2$ queries\footnote{More precisely, $3$ with probability $1-p_0$, and $k+2$ with probability $p_0$, for an expected number $(k-1)p_0+3\simeq k/2$.} to the oracle for $\D$. Moreover,
\[
  \totalvardist{ \tilde{\D} }{ \uniform } \leq (1-p_0)\totalvardist{\D}{\uniform} + p_0 \totalvardist{ \D^{(k)} }{\uniform} \leq (1-p_0)\eps + p_0 2^{k-1}\eps^k
  \leq \frac{\eps}{2} + \left(\frac{1}{4}+\eps^2\right)(2\eps)^k \leq \frac{\eps}{2} + \frac{1}{2}(2\eps)^k
\]
while
\[
  \totalvardist{ \tilde{\D} }{\D} \leq p_0\totalvardist{ \D^{(k)} }{\D}  \leq p_0\left( \eps+2^{k-1}\eps^k \right)\leq \left(\frac{1}{2}+2\eps^2\right)\left( \eps+2^{k-1}\eps^k \right)
  \leq \frac{\eps}{2} + 2\eps^3 + \frac{1}{2}(2\eps)^k
\]
(recalling for the rightmost step of each inequality that $\eps \leq\frac{1}{2}$). Taking $k=3$, one obtains, with a sample complexity at most $5$, a distribution $\tilde{\D}$ satisfying
\begin{align*}
  \totalvardist{ \tilde{\D} }{ \uniform } &\leq \frac{\eps}{2} + 4\eps^3, \qquad
  \totalvardist{ \tilde{\D} }{\D} \leq \frac{\eps}{2}+6\eps^3 \;.
\end{align*}
(Note that assuming $\eps < 1/4$, one can get the more convenient -- yet looser -- bounds $\totalvardist{ \tilde{\D} }{ \uniform } \leq \frac{21}{32}\eps < \frac{2\eps}{3}$, $\totalvardist{ \tilde{\D} }{\D} \leq \frac{97\eps}{128} < \frac{4\eps}{5}$.)
\end{proof}
\fi 
\begin{restatable}[Bootstrapping Improver]{theorem}{unifconvolviiibootstrap}\label{lemma:sampling:corrector:uniformity:boostrap}
  For any $\eps\leq \frac{1}{2}$, $0<\eps_2<\eps$ and $\eps_1=\eps-\eps_2+\bigO{\eps^3}$, there exists a sampling improver for uniformity with query complexity $\bigO{\frac{\log^2\frac{1}{\eps_2}}{\log\frac{1}{\eps}}}$.
\end{restatable}
\ifnum\fulldetails=1 \begin{proof}
We show how to obtain such a guarantee -- note that the constant $27$ in the $\bigO{\eps^3}$ is not tight, and can be reduced at the price of a more cumbersome analysis. Let $\alpha > 0$ be a parameter (to be determined later) satisfying $\alpha < \eps^2$, and $k$ be the number of bootstrapping steps -- i.e., the number of time one recursively apply the construction of \autoref{lemma:sampling:corrector:uniformity:hybrid} with $\alpha$. We write $\D_j$ for the distribution obtained after the $j$\textsuperscript{th} recursive step, so that $\D_0=\D$ and $\hat{\D}=\D_k$; and let $u_j$ (resp. $v_j$) denote an upper bound on $\totalvardist{\D_j}{\uniform}$ (resp. $\totalvardist{\D_j}{\D}$). Note that by the guarantee of \autoref{lemma:sampling:corrector:uniformity:hybrid} and applying a triangle inequality for $v_j$, one gets the following recurrence relations for $(u_j)_{0\leq j\leq k}$ and $(v_j)_{0\leq j\leq k}$:
\begin{align*}
  u_0 &= \eps,\qquad u_{j+1}=\frac{1}{2}u_j + \alpha\\
  v_0 &= 0,\qquad v_{j+1}=\left(\frac{1}{2}u_j + 2u_j^3+\alpha\right) + v_j
\end{align*}
Solving this recurrence for $u_k$ gives
\begin{equation}\label{eq:sampling:corrector:uniformity:boostrap:uk}
  u_k = \frac{\eps}{2^k}+2\left( 1-\frac{1}{2^k} \right)\alpha < \frac{\eps}{2^k}+2\alpha
\end{equation}
while one gets an upper bound on $v_k$ by writing
\begin{align*}
  v_k &= v_k -v_0 = \sum_{j=0}^{k-1} \left( v_{j+1} - v_j \right) = k\alpha + \frac{1}{2}\sum_{j=0}^{k-1} u_j + 2\sum_{j=0}^{k-1}u_j^3 \\
  &= 2k\alpha + \left( 1-\frac{1}{2^k} \right)\eps - 2\left( 1-\frac{1}{2^k} \right)\alpha + 2\sum_{j=0}^{k-1}u_j^3 \\
  &< \left( 1-\frac{1}{2^k} \right)\eps + 2\underbrace{\left( k-1+\frac{1}{2^k} \right)\alpha}_{\leq k\alpha} + \left( 3\eps^3+16\eps^2\alpha+48\eps\alpha^2+16k\alpha^3 \right)
\end{align*}
where we used the expression \eqref{eq:sampling:corrector:uniformity:boostrap:uk} for $u_j$. Since $\alpha < \eps^2 \leq \frac{1}{4}$, we can bound the rightmost terms as $16k\alpha^3 \leq k\alpha$, $48\eps\alpha^2<48\eps^5$ and $16\eps^2\alpha < 16\eps^4$, so that 
\begin{equation}\label{eq:sampling:corrector:uniformity:boostrap:vk}
  v_k < \left( 1-\frac{1}{2^k} \right)\eps + 3k\alpha + 3\eps^3+16\eps^4+48\eps^5 < \left( 1-\frac{1}{2^k} \right)\eps + 3k\alpha + 23\eps^3
\end{equation}
It remains to choose $k$ and $\alpha$; to get $u_k \leq \eps_2$, set $k\eqdef\clg{\log\frac{\eps}{\eps_2(1-\eps^2)}}\leq \log\frac{4\eps}{3\eps_2}+1$ and $\alpha\eqdef\frac{1}{2}\eps_2\eps^2$, so that $\frac{\eps}{2^k}\leq (1-\eps^2)\eps_2$ and $2\alpha \leq \eps^2\eps_2$. Plugging these values in \eqref{eq:sampling:corrector:uniformity:boostrap:vk},
\begin{align*}
  v_k &\operatorname*{<}_{(\eps_2<\eps)} \left( 1-\frac{\eps_2(1-\eps^2)}{\eps} \right)\eps + \frac{3}{2}k\eps_2\eps^2 + 23\eps^3 
    = \eps - \eps_2 + \frac{3}{2}k\eps_2\eps^2 + 24\eps^3 
    < \eps - \eps_2 + 27\eps^3
\end{align*}
where the last inequality comes from the fact that $\frac{3}{2}k\frac{\eps_2}{\eps} \leq \frac{3}{2}\log\frac{8\eps}{3\eps_2}\cdot\frac{\eps_2}{\eps}\leq 3$.
Therefore, we have $\totalvardist{\D_k}{\uniform}\leq \eps_2$, $\totalvardist{\D_k}{\D} \leq \eps - \eps_2 + 27\eps^3$ as claimed. We turn to the number $m$ of queries made along those $k$ steps; from \autoref{lemma:sampling:corrector:uniformity:hybrid}, this is at most
\begin{align*}
  m \leq \sum_{j=0}^{k-1} \clg{ \frac{\log\frac{1}{\alpha}-1}{\log\frac{1}{u_j}-1} } \leq k\cdot \clg{ \frac{\log\frac{1}{\alpha}-1}{\log\frac{1}{\eps}-1} }= \bigO{ \frac{\log^2\frac{1}{\eps_2}}{\log\frac{1}{\eps}} }
\end{align*}
which concludes the proof.
\end{proof}
\fi

Note that in all four cases, as our improvers do not use any randomness of their own, 
they always output according to the same improved distribution: that is, after fixing 
the parameters $\eps, \eps_2$ and the unknown distribution $\D$, then $\hat{\D}$ is 
uniquely determined, even across independent calls to the improver.

\subsubsection{Correcting uniformity on a subgroup}\label{sec:unif:subgroup}
\paragraph{Outline} It is easy to observe that all the results above still hold when replacing $\Z_n$ by any finite Abelian group $G$. Thus, a natural question to turn to is whether one can generalize these results to the case where the unknown distribution is close to the uniform distribution on an arbitrary, unknown, \emph{subgroup} $H$ of the domain $G$.

\noindent To do so, a first observation is that if $H$ were known, and if furthermore a constant (expected) fraction of the samples were to fall within it, then one could directly apply our previous results by conditioning samples on being in $H$, using rejection sampling. The results of this section show how to achieve this ``identification'' of the subgroup with only a $\log(1/\eps)$ overhead in the sample complexity. At a high-level, the idea is to take a few samples, and argue that their greatest common divisor will (with high probability) be a generator of the subgroup.\medskip

\paragraph{Details} Let $G$ be a finite cyclic Abelian group of order $n$, and $H\subseteq G$ a subgroup of order $m$. We denote by $\uniform_H$ the uniform distribution on this subgroup. Moreover, for a distribution $\D$ over $G$, we write $\D_H$ for the conditional distribution it induces on $H$, that is 
\[
  \forall x\in G,\quad \D_H(x) = \frac{\D(x)}{\D(H)}\indicSet{H}(x)
\]
which is defined as long as $\D$ puts non-zero \new{weight} on $H$. The following lemma shows that if $\D$ is close to $\uniform_H$, then so is $\D_H$:

\begin{restatable}{lemma}{closeuniformsubgroup}\label{lemma:eps:close:subgroup:eps:close:conditional}
Assume $\totalvardist{\D}{\uniform_H} < 1$. Then $\totalvardist{\D_H}{\uniform_H} \leq \totalvardist{\D}{\uniform_H}$.
\end{restatable}
\ifnum\fulldetails=1 \begin{proof}
  First, observe that the assumption implies $\D_H$ is well-defined: indeed, as 
  $\totalvardist{\D}{\uniform_H} = \sup_{S\subseteq G}\left( \uniform_H(S)-\D(S) \right)$, 
  taking $S=H$ yields $1 > \totalvardist{\D}{\uniform_H} \geq \uniform_H(H)-\D(H) = 1-\D(H)$, 
  and thus $D(H) > 0$.

  Rewriting the definition of $\totalvardist{\D_H}{\uniform_H}$, one gets $\totalvardist{\D}{\uniform_H} = \frac{1}{2}\left( \sum_{x\in H} \abs{ \D(x) - \frac{1}{\abs{H}} } + \sum_{x\notin H} \D(x)\right)$; so that
  \begin{align*}
    2\totalvardist{\D_H}{\uniform_H} = \sum_{x\in H} \abs{ \D_H(x) - \frac{1}{\abs{H}} } 
    &\leq \sum_{x\in H} \abs{ \D_H(x) - \D(x) }  + \sum_{x\in H} \abs{ \D(x) - \frac{1}{\abs{H}} } \\
    &= \sum_{x\in H} \abs{ \D_H(x) - \D(x) }  + \left(2\totalvardist{\D}{\uniform_H} - \sum_{x\notin H} \D(x)\right) \\
    &= \sum_{x\in H} \D(x)\abs{ \frac{1}{\D(H)} - 1 }  + 2\totalvardist{\D}{\uniform_H} - \left(1-\D(H)\right) \\
    &= \D(H)\abs{ \frac{1}{\D(H)} - 1 }  + 2\totalvardist{\D}{\uniform_H} - \left(1-\D(H)\right) \\
    &= \abs{1-\D(H)} + 2\totalvardist{\D}{\uniform_H} - \left(1-\D(H)\right)\\
    &= 2\totalvardist{\D}{\uniform_H}.
  \end{align*}
\end{proof}
\fi 
Let $\D$ be a distribution on $G$ promised to be \eps-close to the uniform distribution $\uniform_H$ on some unknown subgroup $H$, for $\eps < \frac{1}{2}-c$. For the sake of presentation, we hereafter without loss of generality identify $G$ to $\mathbb{Z}_n$. Let $h$ be the generator of $H$ with smallest absolute values (when seen as an integer), so that $H=\{0, h,2h,3h,\dots,(m-1)h\}$.

\noindent Observe that $\D(H) > 1-2\eps$, as $2\totalvardist{\D}{\uniform_H} = \sum_{x\in H} \abs{ \D(x) - \frac{1}{\abs{H}} } + \D(H^c)$; therefore, \emph{if} $H$ were known one could efficiently simulate sample access to $\D_H$ via rejection sampling, with only a constant factor overhead (in expectation) per sample. It would then become possible, as hinted in the foregoing discussion, to correct uniformity on $\D_H$ (which is \eps-close to $\uniform_H$ by \autoref{lemma:eps:close:subgroup:eps:close:conditional}) via one of the previous algorithms for Abelian groups. The question remains to show how to find $H$; or, equivalently, $h$.\medskip

\begin{algorithm}
  \begin{algorithmic}
    \Require $\eps\in(0,\frac{1}{2}-c]$, $\SAMP_D$ with $D$ \eps-close to uniform on some subgroup $H\subseteq\Z_n$
    \Ensure Outputs a generator $\hat{h}$ of $H$ with probability $1-\tildeO{\eps}$
    \State Draw $k$ independent samples $s_1,\dots,s_k$ from $D$, for $k\eqdef\bigO{\log\frac{1}{\eps}}$
    \State Compute $\hat{h}=\gcd(s_1,\dots,s_k)$
    \State \Return $\hat{h}$
  \end{algorithmic}\caption{\label{algo:find:generator:subgroup} Algorithm {\sc Find-Generator-Subgroup}}
\end{algorithm}

\begin{restatable}{lemma}{unifsubgroupfindgenerator}\label{lemma:find:generator:subgroup}
  Let $G$, $H$ be as before. There exists an algorithm (\autoref{algo:find:generator:subgroup}) which, given $\eps < 0.49$ as well as sample access to some distribution $D$ over $G$, makes $\bigO{\log\frac{1}{\eps}}$ calls to the oracle and returns an element of $G$. Further, if $\totalvardist{\D}{\uniform_H} \leq \eps$, then with probability at least $1-\tildeO{\eps}$ its output is a generator of $H$.
\end{restatable}
\ifnum\fulldetails=1 
\begin{proof}
In order to argue correctness of the algorithm, we will need the following well-known facts:

\begin{fact}\label{fact:random:integers:coprime}
Fix any $k\geq 1$, and let $p_{n,k}$ be the probability that $k$ independent numbers drawn uniformly at random from $[n]$ be relatively prime. Then 
$
p_{n,k}\xrightarrow[n\to\infty]{} \frac{1}{\zeta(k)}
$
(where $\zeta$ is the Riemann zeta function).
\end{fact}
\begin{fact} One has
$
\zeta(x) {\displaystyle\operatorname*{=}_{x\to\infty}} 1 + \frac{1}{2^x} + \littleO{\frac{1}{2^x}} 
$;
and in particular
$
\frac{1}{\zeta(k)} {\displaystyle\operatorname*{=}_{k\to\infty}} 1 - \frac{1}{2^k} + \littleO{\frac{1}{2^k}}.
$
\end{fact}

\noindent With this in hand, let $k\eqdef\bigO{\log\frac{1}{\eps}}$,  chosen so that $\eps k = \bigTheta{\frac{1}{2^k}} = \bigTheta{\frac{\eps}{\log\frac{1}{\eps}}}$. We break the analysis of our subgroup-finding algorithm in two cases:
\paragraph*{Case 1: $\abs{H}=\bigTheta{1}$} This is the easy case: if $H$ only contains constantly many elements ($m$ is a constant of $n$ and $\eps$), then after taking $k$ samples $s_1,\dots,s_k\sim D$, we have
\begin{itemize}
  \item $s_1,\dots,s_k\in H$ (event $E_1$) with probability at least $(1-\eps)^k = 1 - \bigO{k\eps}$;  \item the probability that there exists an element of $H$ not hit by any of the $s_i$'s is at most, by a union bound,
    \[
       \sum_{x\in H} (1-\D(x))^k \leq m\left(1-\frac{1}{m}+\eps\right)^k = 2^{-\bigOmega{k}}
    \]
    for $\eps$ sufficiently small ($\eps \ll \frac{1}{m}$). Let $E_2$ be the event each element of $H$ appears amongst the samples.
\end{itemize}
Overall, with probability $1 - \bigO{k\eps}$ (conditioning on $E_1$ and $E_2$), our set of samples is exactly $H$, and $\gcd(s_1,\dots,s_k)=\gcd(H)=h$.
\paragraph*{Case 2: $\abs{H}=\omega(1)$} This amounts to saying that $h=\littleO{n}$. In this case, taking again $k$ samples $s_1,\dots,s_k\sim D$ and denoting by $\hat{h}$ their greatest common divisor:
\begin{itemize}
  \item $s_1,\dots,s_k\in H$ (event $E_1$) with probability at least $(1-\eps)^k = 1 - \bigO{k\eps}$ as before;
  \item conditioned on $E_1$, note that if the $s_i$'s were \emph{uniformly} distributed in $H$, then the probability that $\hat{h}=h$ would be exactly $p_{\frac{n}{h},k}$ -- as $\gcd(ha,\dots,hb)=h$ if and only if $\gcd(a,\dots,b) =1$, i.e. if $a,\dots,b$ are relatively prime. In this ideal scenario, therefore, we would have
  \[
  \probaCond{ \gcd(s_1,\dots,s_k) = h }{E_1} = p_{\frac{n}{h},k} \xrightarrow[n\to\infty]{} \frac{1}{\zeta(k)} = 1-\bigO{\frac{1}{2^k} }
  \]
  by \autoref{fact:random:integers:coprime} and our assumption $h=\littleO{n}$.
  
  \noindent To adapt this result to our case -- where $s_1,\dots,s_k\sim D_H$ (as we conditioned on $E_1$), it is sufficient to observe that by the Data Processing Inequality for total variation distance,
  \[
  \abs{ \probaDistrOf{s_1,\dots,s_k\sim D_H}{ \gcd(s_1,\dots,s_k) = h } - \probaDistrOf{s_1,\dots,s_k\sim \uniform_H}{ \gcd(s_1,\dots,s_k) = h } } \leq \totalvardist{\D_H^{\otimes k}}{\uniform_H^{\otimes k}} \leq k\eps
  \]
  so that in our case 
  \begin{equation}
  \probaCond{ \gcd(s_1,\dots,s_k) = h }{E_1} \geq p_{\frac{n}{h},k} - k\eps \xrightarrow[n\to\infty]{} \frac{1}{\zeta(k)} - k\eps = 1-\bigO{\frac{1}{2^k}} = 1-\bigO{\frac{\eps}{\log\frac{1}{\eps}}}
  \end{equation}
\end{itemize}

In either case, with probability at least $1-\tildeO{\eps}$, we find a generator $h$ of $H$, acting as a (succinct) representation of $H$ which allows us to perform rejection sampling.
\end{proof}

\fi 
This directly implies the theorem below: any sampling improver for uniformity on a group directly yields an improver for uniformity on an unknown \emph{subgroup}, with essentially the same complexity.
\begin{restatable}{theorem}{unifsubgroup}\label{lemma:correcting:subgroup:rejection:sampling}
Suppose we have an $(\eps,\eps_1,\eps_2)$-sampling improver for uniformity over Abelian finite cyclic groups, with query complexity $q(\eps,\eps_1,\eps_2)$.
Then there exists an $(\eps,\eps_1,\eps_2)$-sampling improver for uniformity on subgroups, with query complexity
\[
\bigO{ \log\frac{1}{\eps} + q(\eps,\eps_1,\eps_2)\log q(\eps,\eps_1,\eps_2) }
\]
\end{restatable}
\begin{proof}Proof is straightforward (rejection sampling over the subgroup, once identified: constant probability of hitting it, so by trying at most $\bigO{\log q}$ draws per samples before outputting \fail, one can provide a sample from $\D_H$ to the original algorithm with probability $1-1/10q$, for each of the (at most) $q$ queries).
\end{proof}

\subsection{Comparison with randomness extractors}\label{sec:unif:extractors}
In the randomness extractor model, 
one is provided with a
source of imperfect random bits (and sometimes an
additional source of completely random bits),
and the goal is to output as many random
bits as possible that are close to uniformly distributed.
In the distribution corrector model, one
is provided with a distribution that is {\em close to
having} a property $\property$, and the goal is to have the ability to
generate a {\em similar}
distribution that {\em has} property $\property$.

One could therefore view extractors as sampling improvers
for the property of uniformity of distributions 
(i.e. $\property=\{\uniform_n\})$: 
indeed, both sampling correctors and extractors attempt to minimize the use of extra
randomness.
However, there are significant differences between the two settings.
A first difference is that randomness extractors assume a lower bound on the 
min-entropy\footnote{The min-entropy of a distribution 
$\D$ is defined as $H_{\infty}(\D)=\log\frac{1}{\max_i{\D(i)}}$.} 
of the input distribution, whereas sampling improvers 
assume the distribution to be $\eps$-close to uniform in total variation distance. 
Note that the two assumptions are not comparable.\footnotemark 
Secondly, in both the extractor and sampling improver
models, since the entropy of the output distribution should be larger, 
one would either need more random bits from the 
weak random source or additional uniform random bits. 
Our sampling improvers do not use any extra random bits, 
which is also the case in deterministic extractors, but not in other
extractor constructions.
However, unlike the extractor model, in the sampling improver model, 
there is no bound on the number of 
independent samples one can take from the original distribution. 
Tight bounds and impossibility results are known for both general and deterministic extractors \cite{book:vadhan,RT:00}, 
in particular in terms of the amount of additional randomness required. 
Because of the aforementioned differences in both the assumptions on and access to the input distribution, 
these lower bounds do not apply to our setting -- which explains why our sampling improvers avoid this need for extra random bits.\medskip 

\footnotetext{For example, the min-entropy of a distribution which is $\eps$-close to uniform 
can range from $\log (1/\eps)$ to $\bigOmega{\log n}$. 
Conversely, the distance to uniformity of a distribution which 
has high min-entropy can also vary significantly: there exist distributions with min-entropy $\bigOmega{\log n}$
but which are respectively $\bigOmega{1}$-far from and $\bigO{1/n}$-close to uniform.}

	\subsection{Monotone distributions and randomness scarcity}\label{sec:monotonicity:extract}
		In this section, we describe how to utilize a (close-to-monotone) input distribution to obtain the uniform random samples some of our previous correctors and improvers need. This is at the price of a $\tildeO{\log n}$-sample overhead per draw, and follows the same general approach as in \autoref{lemma:sampling:corrector:uniformity:vneumann}. We observe that even if this \emph{seems} to defeat the goal (as, with this many samples, one could even \emph{learn} the distribution, as stated in \autoref{lemma:learning:corrector:birge}), this is not actually the case: indeed, the procedure below is meant as a subroutine for these very same correctors, emancipating them from the need for truly additional randomness -- which they would require otherwise, e.g. to generate samples from the corrected or learnt distribution.
\begin{restatable}[Randomness from (almost) monotone]{lemma}{monotonevncorr}\label{lemma:randomness:monotone:vneumann}
  There exists a procedure which, which, given $\eps\in[0,1/3)$ and $\delta > 0$, as well as sample access to a distribution $\D$ guaranteed to be \eps-close to monotone, either returns
  \begin{itemize}
    \item \textsf{``point mass,''} if $\D$ is $\eps$-close to the point distribution\footnote{The Dirac distribution $\delta_1$ defined by $\delta_1(1)=1$, which can be trivially sampled from without the use for any randomness by always outputting 1.}  on the first element;
    \item or a uniform random sample from $[n]$;
  \end{itemize}
   with probability of failure at most $\delta$. The procedure makes $\bigO{\frac{1}{\eps^2} \log\frac{1}{\delta}}$ samples from $\D$ in the first case, and $\bigO{\frac{\log n}{\eps} \log\frac{\log n}{\delta}}$ in the second.
\end{restatable}
\ifnum\fulldetails=1 \begin{proof}
By taking $\bigO{\log(1/\delta)/\eps^2}$ samples, the algorithm starts by approximating by $\hat{F}$ the \cdf $F$ of the distribution up to an additive $\frac{\eps}{4}$ in $\lp[\infty]$. Then, defining
\[
    m \eqdef \min \setOfSuchThat{i \in [n]}{ \hat{F}(i) \geq 1-\frac{\eps}{2}}
\]
so that $F(m) \geq 1-\frac{3\eps}{4}$. According to the value of $m$, we consider two cases:
\begin{itemize}
  \item If $m = 1$, then $\D$ is \eps-close to the (monotone) distribution $\delta_1$ which has all weight on the first element; this means we have effectively \emph{learnt} the distribution, and can thereafter consider, for all purposes, $\delta_1$ in lieu of $\D$.  \item If $m > 1$, then $\D(1) < 1-\frac{\eps}{4}$ and we can partition the domain in two sets $S_0\eqdef\{1,\dots, k\}$ and $S_1\eqdef\{k,\dots, n\}$, by setting
\[
    k \eqdef \min \setOfSuchThat{i \in [n]}{ \hat{F}(i) < 1-\frac{\eps}{2}}
\]
\end{itemize}

By our previous check we know that this quantity is well-defined. Further, this implies that $D(S_0) < 1-\frac{\eps}{4}$ and $\D(S_1) >\frac{\eps}{4}$  (the actual values being known up to $\pm \frac{\eps}{4}$). $\D$ being \eps-close to monotone, it also must be the case that
\[
  \D(S_0) \geq \frac{k}{k+1}\left(1-\frac{3\eps}{4}\right) - 2\eps \geq \frac{1}{2}\left(1-\frac{3\eps}{4}\right) - 2\eps \geq \frac{1}{2}-\frac{19\eps}{8}
\]
since $\hat{F}(k+1) \geq 1-\frac{\eps}{2}$ implies $\D(\{1,\dots,k\})+\D(k)=\D(\{1,\dots,k+1\}) \geq 1-\frac{3\eps}{4}$. By setting $p\eqdef \D(S_0)$, this means we have access to a Bernoulli random variable with parameter $p\in\left[\frac{1}{2}-3\eps,1-\frac{\eps}{4}\right]$. As in the proof of \autoref{lemma:sampling:corrector:uniformity:vneumann} (the constant $c$ being replaced by $\min\!\left(\frac{\eps}{4}, \frac{1}{2}-3\eps\right)$), one can then leverage this to output with probability at least $1-\delta$ a uniform random number $s\in[n]$ using $\bigO{\frac{\log n}{\eps} \log\frac{\log n}{\delta}}$ samples---and $\bigO{\frac{\log n}{\eps}}$ in expectation.
\end{proof}
\fi  
\clearpage
\clearpage
\addcontentsline{toc}{section}{References}
\bibliographystyle{alpha}	
\bibliography{references}

\clearpage
\appendix
\section{On convolutions of distributions over an Abelian finite cyclic group}\label{appendix:convolution:abelian}

\begin{definition}
For any two probability distributions $\D_1$, $\D_2$ over a finite group $G$ (not necessarily Abelian), the \emph{convolution} of $\D_1$ and $\D_2$, denoted $\D_1\convolution \D_2$, is the distribution on $G$ defined by
\[
  \D_1\convolution \D_2(x) = \sum_{g\in G} \D_1(x g^{-1})\D_2(g)
\]
In particular, if $G$ is Abelian, $\D_1\convolution \D_2=\D_2\convolution \D_1$.
\end{definition}
\begin{fact}
The convolution satisfies the following properties:
\begin{enumerate}[(i)]
  \item it is associative:
    \begin{equation}
      \forall \D_1,\D_2,\D_3, \quad \D_1\convolution (D_2\convolution \D_3) = (D_1\convolution \D_2)\convolution \D_3 = \D_1\convolution \D_2\convolution \D_3
    \end{equation}
  \item it has a (unique) absorbing element, the uniform distribution $\uniformOn{G}$:
    \begin{equation}
      \forall \D, \quad \uniformOn{G}\convolution \D = \uniformOn{G}
    \end{equation}
  \item it can only decrease the total variation:
    \begin{equation}
      \forall \D_1,\D_2,\D_3, \quad \totalvardist{ \D_1\convolution \D_2 }{ \D_1\convolution \D_3 } \leq \totalvardist{ \D_2 }{ \D_3 }
    \end{equation}
  \end{enumerate}
\end{fact}
\noindent For more on convolutions of distributions over finite groups, see for instance \cite{Diaconis:88} or \cite{BCLR:08}.

\begin{fact}[\cite{MathOverflow:13:main}]\label{fact:convol:uniform:tv}
Let $G$ be a finite Abelian group, and $\D_1$, $\D_2$ two probability distributions over $G$. Then, the convolution of $\D_1$ and $\D_2$ satisfies
\begin{equation}\label{eq:convol:uniform:tv}
  \totalvardist{ \uniformOn{G} }{ \D_1\convolution \D_2 } \leq 2\totalvardist{ \uniformOn{G} }{ \D_1 }\totalvardist{ \uniformOn{G} }{ \D_2 }
\end{equation}
where $\uniformOn{G}$ denotes the uniform distribution on $G$. Furthermore, this bound is tight.
\end{fact}

\begin{claim}\label{claim:counterexample:convolution}
For $\eps\in(0,\frac12)$, there exists a distribution $\D$ on $\Z_n$ such that $\totalvardist{ \D }{ \uniform } = \eps$, yet $\totalvardist{ \D }{ \D\convolution \D } = \eps + \frac{3}{4}\eps^2 + \bigO{\eps^3} > \eps$.
\end{claim}
\begin{proof}
Inspired by---and following---a question on MathOverflow~(\cite{MathOverflow:13:counterexample}). Setting $\delta = 1-\eps > 1/2$, and taking $\D_A$ to be uniform on a subset $A$ of $\Z_n$ with $\abs{A}=\delta n$, one gets that $\totalvardist{\D_A}{\uniform}=\eps$, and yet
  \[
    \totalvardist{ \D_A }{ \D_A\convolution \D_A } = \frac{1}{2}\normone{ \D_A - \D_A\convolution \D_A } = \frac{1}{2} \sum_{g\in G} \abs{ \frac{\mathds{1}_A(g)}{\abs{A}^{\vphantom{2}}}-\frac{r(g)}{\abs{A}^2} } = 1-\frac{1}{\abs{A}^2}\sum_{a\in A}r(a)
  \]
  where $r(g)$ is the number of representations of $g$ as a sum of two elements of $A$ (as one can show that $\D_A\convolution \D_A(g)=\abs{A}^{-2}r(g)$). Fix $A$ to be the interval of length $\delta n$ centered around $n/2$, that is $A=\{\ell,\dots,L\}$ with 
  \[
      \ell \eqdef \frac{1-\delta}{2}n, \qquad L \eqdef \frac{1+\delta}{2}n
  \]
  Computing the quantity $\sum_{a\in A}r(a)$ amounts to counting the number of pairs $(a,b)\in A\times A$ whose sum (modulo $n$) lies in $A$. For convenience, define $k = (1-\delta)n$ and $K=\delta n$:
  \begin{itemize}
    \item for $k \leq a \leq K$, $a+b$ ranges from $\ell+a \leq L$ to $L+a \geq \ell + n$, so that modulo $n$ exactly $\abs{A} - \abs{A^c} =  (2\delta-1)n$ elements of $A$ are reached (each of them exactly once);
    \item for $\ell \leq a < k$, $a+b$ ranges from $\ell+a < L$ to $L+a < \ell + n$, so that the elements of $A$ not obtained are those in the interval $\{\ell, \ell+a-1\}$ -- there are $a$ of them -- and again the others are obtained exactly once;
    \item for $K < a \leq L$, $a+b$ ranges from $L < \ell+a \leq n$ to $L+a \leq K + n$, so that the elements of $A$ not obtained are those in the interval $\{L+a-n+1, L\}$ -- there are $n-a$ of them -- and as before the others are hit exactly once.
  \end{itemize}
  It follows that
  \begin{align*}
    \sum_{a\in A}r(a) &= (K-k+1)(2\delta-1)n + \sum_{a=\ell}^{k-1} (\delta n-a) + \sum_{a=K+1}^{L} (\delta n - (n-a)) \\
     &= (2\delta-1)(2\delta-1)n^2 + \sum_{a=\ell}^{k-1} (\delta n-a) + \sum_{a=K+1}^{L} (a-(1-\delta) n)  \tag{the 2 sums are equal}\\
     &= (2\delta-1)^2n^2 + 2\cdot\frac{n^2}{8}(7\delta-3)(1-\delta) + \bigO{n} \\
    &= \frac{1}{4}\left( 4(4\delta^2-4\delta+1)+(7\delta-3)(1-\delta) \right)n^2 + \bigO{n} = \frac{1}{4}\left( 9\delta^2-6\delta+1\right)n^2 + \bigO{n} \\
    &= \left(1-3\eps+\frac{9}{4}\eps^2\right)n^2 + \bigO{n}
  \end{align*}
  and thus
  \begin{align*}
    \totalvardist{ \D_A }{ \D_A\convolution \D_A } = 1-\frac{1}{\abs{A}^2}\sum_{a\in A}r(a) &= 1- \frac{1-3\eps+\frac{9}{4}\eps^2}{(1-\eps)^2} + \bigO{\frac{1}{n}} = \eps+\frac{3}{4}\eps^2 + \bigO{\eps^3}
  \end{align*}
  (as $\eps = \omega(1/\sqrt[3]{n})$).
\end{proof}
 
\section{Formal definitions: learning and testing}\label{appendix:definitions}
In this appendix, we define precisely the notions of \emph{testing}, \emph{tolerant testing}, \emph{learning} and \emph{proper learning} of distributions over a domain $[n]$.
\begin{definition}[Testing]\label{def:testing:alg}
  Fix any property $\property$ of distributions, and let $\ORACLE_\D$ be an oracle providing some type of access to $D$. A \emph{$q$-sample testing algorithm for $\property$} is a randomized algorithm $\Tester$ which takes as input $n$, $\eps\in(0,1]$, as well as access to $\ORACLE_\D$. After making at most $q(\eps,n)$ calls to the oracle, \Tester outputs either \accept or \reject, such that the following holds:
  \begin{itemize}
    \item if $\D\in\property$, \Tester outputs \accept with probability at least $2/3$;
    \item if $\totalvardist{\D}{\property} \geq \eps$, \Tester outputs \reject with probability at least $2/3$.
  \end{itemize}
\end{definition}
\noindent We shall also be interested in \emph{tolerant} testers -- \newer{roughly}, algorithms robust to a relaxation of the first item above:
\begin{definition}[Tolerant testing]\label{def:tol:testing:alg}
  Fix property $\property$ and $\ORACLE_\D$ as above. A \emph{$q$-sample tolerant testing algorithm for $\property$} is a randomized algorithm $\Tester$ which takes as input $n$, $0 \leq \eps_1 < \eps_2 \leq 1$, as well as access to $\ORACLE_\D$. After making at most $q(\eps_1,\eps_2,n)$ calls to the oracle, \Tester outputs either \accept or \reject, such that the following holds:
  \begin{itemize}
    \item if $\totalvardist{\D}{\property} \leq \eps_1$, \Tester outputs \accept with probability at least $2/3$;
    \item if $\totalvardist{\D}{\property} \geq \eps_2$, \Tester outputs \reject with probability at least $2/3$.
  \end{itemize}
\end{definition}

Note that the above definition is quite general, and can be instantiated for different types of oracle access to the unknown distribution. In this work, we are mainly concerned with two such settings, namely the \emph{sampling oracle access} and \emph{\Cdfsamp oracle access}:
  \begin{definition}[Standard access model (sampling)]\label{def:sampling:oracle}
      Let $\D$ be a fixed distribution over $[n]$. A \emph{sampling oracle for $\D$} is an oracle $\SAMP_\D$ defined as follows: when queried, $\SAMP_\D$ returns an element $x\in[n]$, where the probability that $x$ is returned is $\D(x)$ independently of all previous calls to the oracle.
  \end{definition}
    \begin{definition}[\Cdfsamp access model \cite{BKR:04,CR:14}]\label{def:extended:oracle}
    Let $\D$ be a fixed distribution over $[n]$. A \emph{\cdfsamp oracle for $\D$} is a pair of oracles $(\SAMP_\D, \CDFEVAL_\D)$ defined as follows: the \emph{sampling} oracle $\SAMP_\D$ behaves as before, while the \emph{evaluation} oracle $\CDFEVAL_\D$ takes as input a query element $j\in[n]$, and returns the value of the cumulative distribution function (\cdf) at $j$. That is, it returns the probability weight that the distribution puts on $[j]$, $\D([j])=\sum_{i=1}^j D(i)$.
    \end{definition}  

\noindent Another class of algorithms we consider is that of \emph{(proper) learners}. We give the precise definition below:  
\begin{definition}[Learning]
Let $\class$ be a class of probability distributions and $\D\in \class$ be an unknown distribution. Let also $\mathcal{H}$ be a hypothesis class of distributions. 
A \emph{$q$-sample learning algorithm for $\class$} is a randomized algorithm $\Learner$ which takes as input $n$, $\eps,\delta\in(0,1)$, as well as access to $\SAMP_\D$ and  outputs the description of a distribution $\hat{\D}\in \mathcal{H}$ such that with probability at least $1-\delta$ one has $\totalvardist{\D}{\hat{\D}}\leq \eps$. 

\noindent If in addition $\mathcal{H}\subseteq \class$, then we say \Learner is a \emph{proper learning algorithm}.   
\end{definition}

The exact formalization of what \emph{learning a probability distribution} means has been considered in  Kearns et al. \cite{Kearns:94}. We note that in their language, the variant of learning this paper is most closely related to is \emph{learning to generate}.
 
\section{Omitted proofs}\label{appendix:misc:proofs}
This section contains the proofs of some of the technical lemmata of the paper, omitted for the sake of conciseness. 
\begin{proof}[Proof of Eq.~\eqref{eq:Birge:tv}]
    Fix a partition $\mathcal{I}$ of $[n]$ into $\ell$ intervals $I_1,\dots, I_\ell$, and let $\D_1$, $\D_2$ be two arbitrary distributions on $[n]$. Recall that $\Psi_{\mathcal{I}}(\D)$ is the flattening of distribution $\D_j$ (with relation to the partition $\mathcal{I}$).
      \begin{align*}
      2\totalvardist{ \Psi_{\mathcal{I}}(\D_1) }{ \Psi_{\mathcal{I}}(\D_2) } &= \sum_{i=1}^n \abs{ \Psi_{\mathcal{I}}(\D_1)(i) - \Psi_{\mathcal{I}}(\D_2)(i) } = \sum_{k=1}^\ell \sum_{i\in I_k} \abs{ \frac{\D_1(I_k)}{\abs{I_k}} - \frac{\D_2(I_k)}{\abs{I_k}} } \\
      &=  \sum_{k=1}^\ell \abs{ \D_1(I_k) - \D_2(I_k) } = \sum_{k=1}^\ell \abs{ \sum_{i\in I_k}\left( \D_1(i) - \D_2(i) \right) } \\
      &\leq  \sum_{k=1}^\ell \sum_{i\in I_k} \abs{ \D_1(i) - \D_2(i)  } = \sum_{i=1}^n \abs{ \D_1(i) - \D_2(i)  } 
      = 2\totalvardist{\D_1}{\D_2}.
      \end{align*}
(we remark that Eq.~\eqref{eq:Birge:tv} could also be obtained directly by applying the data processing inequality for total variation distance (\autoref{lemma:data:processing:inequality:total:variation}) to $\D_1$, $\D_2$, for the transformation $\Psi_{\mathcal{I}}(\cdot)$.)
\end{proof}
    
\begin{proof}[Proof of \autoref{coro:Birge:decomposition:robust}]
  Let $\D$ be \eps-close to monotone, and $D^\prime$ be a monotone distribution such that $\totalvardist{\D}{\D^\prime} = \eta \leq \eps$. By Eq.~\eqref{eq:Birge:tv}, we have
  \begin{equation}\label{eqn:phialpha:cannot:increase:tv}
      \totalvardist{ \birge[\D]{\alpha} }{ \birge[\D^\prime]{\alpha} } \leq \totalvardist{\D}{\D^\prime} = \eta
  \end{equation}
  proving the last part of the claim  (since $\birge[\D^\prime]{\alpha}$ is easily seen to be monotone).

  \noindent Now, by the triangle inequality,
  \begin{align*}
      \totalvardist{D}{ \birge[\D^\prime]{\alpha} } &\leq \totalvardist{\D}{\D^\prime} + \totalvardist{D^\prime}{ \birge[\D^\prime]{\alpha} } + \totalvardist{ \birge[\D^\prime]{\alpha} }{ \birge[\D]{\alpha} } \\
      &\leq \eta + \alpha + \eta \\
      &\leq 2\eps+\alpha
  \end{align*}
  where the last inequality uses the assumption on $\D^\prime$ and \autoref{theorem:Birge:obl:decomp} applied to it.
\end{proof}

\end{document}